\documentclass[sigconf, nonacm]{acmart}

\newcommand\vldbdoi{XX.XX/XXX.XX}
\newcommand\vldbpages{XXX-XXX}
\newcommand\vldbvolume{15}
\newcommand\vldbissue{X}
\newcommand\vldbyear{2022}
\newcommand\vldbauthors{\authors}
\newcommand\vldbtitle{\shorttitle} 
\newcommand\vldbavailabilityurl{URL_TO_YOUR_ARTIFACTS}
\newcommand\vldbpagestyle{plain}

\usepackage{amsmath}

\usepackage{xcolor}
\usepackage{graphdb}
\usepackage{booktabs}

\makeatletter
\newenvironment{customlegend}[1][]{%
    \begingroup
    \pgfplots@init@cleared@structures
    \pgfplotsset{#1}%
}{%
    \pgfplots@createlegend
    \endgroup
}%

\def\addlegendimage{\pgfplots@addlegendimage}
\makeatother

\title{Time- and Space-Efficient Regular Path Queries on Graphs}

\author{Diego Arroyuelo}
\orcid{0000-0002-2509-8097}
\affiliation{%
  \institution{Universidad Técnica Federico Santa
  María \& IMFD}
  \city{Santiago}
  \state{Chile}
}
\email{darroyue@inf.utfsm.cl}

\author{Aidan Hogan}
\orcid{0000-0001-9482-1982}
\affiliation{%
  \institution{DCC, University of Chile \& IMFD}
  \city{Santiago}
  \state{Chile}
}
\email{ahogan@dcc.uchile.cl}

\author{Gonzalo Navarro}
\affiliation{%
  \institution{DCC, University of Chile \& IMFD}
  \city{Santiago}
  \state{Chile}
}
\email{gnavarro@dcc.uchile.cl}

\author{Javiel Rojas-Ledesma}
\affiliation{%
  \institution{DCC, University of Chile \& IMFD}
  \city{Santiago}
  \state{Chile}
}
\email{jrojas@dcc.uchile.cl}

\newtheorem{fact}{Fact}
\newcommand{\dd}{\mathinner{.\,.}}
\newcommand{\LF}{\mathsf{LF}}
\newcommand{\rank}{\mathsf{rank}}

\renewcommand{\C}{\ensuremath{\Sigma}}
\newcommand{\V}{\ensuremath{\Phi}}

\newcommand{\dom}[1]{\ensuremath{\mathsf{dom}(#1)}}

\newcommand{\inv}[1]{\ensuremath{\textsc{\^{}}#1}}
\newcommand{\comp}[1]{\ensuremath{#1^{\leftrightarrow}}}

\newcommand{\Ci}{\inv{\C}}
\newcommand{\Cc}{\comp{\C}}

\newcommand{\no}[1]{}

\newcommand{\lab}[1]{\ensuremath{\mathsf{word}(#1)}}

\newcommand{\rcat}[2]{\ensuremath{#1/#2}}
\newcommand{\ror}[2]{\ensuremath{#1|#2}}
\newcommand{\rkle}[1]{\ensuremath{#1^{*}}}
\newcommand{\rplu}[1]{\ensuremath{#1^{+}}}
\newcommand{\ropt}[1]{\ensuremath{#1^{?}}}
\newcommand{\rinv}[1]{\inv{#1}}

\newcommand{\gf}[1]{\textsf{#1}\xspace}
\newcommand{\lone}{\gf{l1}}
\newcommand{\ltwo}{\gf{l2}}
\newcommand{\lfive}{\gf{l5}}

\newcommand{\metro}{\rplu{(\ror{\ror{\lone}{\ltwo}}{\lfive})}}

\newenvironment{example}{\bigskip\noindent {\em Example.}}{\hfill $\Box$\\}

\begin{document}

\begin{abstract}
We introduce a time- and space-efficient technique to solve regular path 
queries over labeled graphs. We combine a bit-parallel simulation of the
Glushkov automaton of the regular expression with the ring index introduced 
by Arroyuelo et al., exploiting its wavelet tree representation 
of the triples in order to efficiently reach the states of the product 
graph that are relevant for the query. Our query algorithm is able to simultaneously process several automaton states, as well as several graph nodes/labels. Our experimental results show 
that our representation uses 3--5 times less space than the alternatives in the literature, while generally outperforming them in query times (1.67 times faster than the next best).
\end{abstract}

\maketitle

\pagestyle{\vldbpagestyle}
\begingroup\small\noindent\raggedright\textbf{PVLDB Reference Format:}\\
\vldbauthors. \vldbtitle. PVLDB, \vldbvolume(\vldbissue): \vldbpages, \vldbyear.\\
\href{https://doi.org/\vldbdoi}{doi:\vldbdoi}
\endgroup
\begingroup
\renewcommand\thefootnote{}\footnote{\noindent
This work is licensed under the Creative Commons BY-NC-ND 4.0 International License. Visit \url{https://creativecommons.org/licenses/by-nc-nd/4.0/} to view a copy of this license. For any use beyond those covered by this license, obtain permission by emailing \href{mailto:info@vldb.org}{info@vldb.org}. Copyright is held by the owner/author(s). Publication rights licensed to the VLDB Endowment. \\
\raggedright Proceedings of the VLDB Endowment, Vol. \vldbvolume, No. \vldbissue\ %
ISSN 2150-8097. \\
\href{https://doi.org/\vldbdoi}{doi:\vldbdoi} \\
}\addtocounter{footnote}{-1}\endgroup

\ifdefempty{\vldbavailabilityurl}{}{
\vspace{.3cm}
\begingroup\small\noindent\raggedright\textbf{PVLDB Artifact Availability:}\\
The source code, data, and/or other artifacts have been made available at \url{https://github.com/darroyue/Ring-RPQ}.
\endgroup
}

\section{Introduction}

A characteristic feature of graph databases is the ability to query over paths of arbitrary length~\cite{AnglesABHRV17}. This feature is typically supported as \textit{regular path queries} (\textit{RPQs})~\cite{CruzMW87,MendelzonW95}, which specify a regular expression that constrains matching paths. Consider the graph of Fig.~\ref{fig:dg2} describing different means of transportation within Santiago de Chile. Edges are directed and labeled with the type of transportation (\lone, \ltwo and \lfive denote three metro lines). We can find pairs of locations reachable by metro with an RPQ $x \xrightarrow{\metro} y$, where $x$ and $y$ are variables over the nodes of the graph, while the regular expression $\metro$ will match paths of length one-or-more such that each edge has the label \lone, \ltwo or \lfive. We may also fix one or both nodes in an RPQ, for example $\gf{Baquedano} \xrightarrow{\metro} y$ finds nodes reachable from $\gf{Baquedano}$ by metro.

While regular path queries have long been studied in theoretical works~\cite{CruzMW87,MendelzonW95}, recently they have been included in practical query languages for graphs~\cite{AnglesABHRV17}. The SPARQL~1.1~\cite{sparql11} query language for RDF graphs includes \textit{property paths}~\cite{KostylevR0V15}, which extend RPQs with inverse paths and negated edge labels. Oracle's graph query language PGQL~\cite{RestHKMC16} also supports RPQs, as does the G-CORE query language~\cite{AnglesABBFGLPPS18} defined by the LDBC standardisation committee. The Cypher query language~\cite{FrancisGGLLMPRS18}, supported by Neo4j, includes limited forms of RPQs (with Kleene-star and concatenation). In summary, RPQs have become a key feature in modern graph databases~\cite{AnglesABHRV17}, and are frequently used: in an analysis of 208 million SPARQL queries issued to the Wikidata Query Service~\cite{MalyshevKGGB18}, Bonifati et al.~\cite{BonifatiMT19} find that 24\% of the queries use at least one RPQ/property path feature.

The problem of efficiently evaluating RPQs has been gaining increasing attention in recent years~\cite{MiaoST07,GubichevN11,KoschmiederL12,DeyC0GWL13,GubichevBS13,YakovetsGG13,WangRJLYF14,FletcherPP16,NoleS16,WangWZ16,YakovetsGG16,Abul-Basher17,BaierDRV17,HartigP17,NguyenK17,ColazzoMNS18,FiondaPC19,MehmoodSSNd19,MiuraAK19,WadhwaPRBB19,JachietGGL20,PacaciBO20,TetzelLK20,GuoGZ21,KuijpersFLY21,LiuWLLW21}. The traditional algorithm -- used, for example, in the theoretical literature to prove complexity bounds -- is based on representing the regular expression of the RPQ as a non-deterministic finite automaton, defining the product graph of the data graph and the automaton, and then applying graph search on the product graph (BFS, DFS, etc.)~\cite{MendelzonW95}. While the product graph is potentially large, practical algorithms can avoid materialising it, and rather expand it lazily during navigation. Other more recent approaches propose the use of recursive queries~\cite{YakovetsGG13,YakovetsGG16,JachietGGL20}, parallel~\cite{MiuraAK19} and distributed~\cite{NoleS16,WangWZ16,ColazzoMNS18,MehmoodSSNd19,GuoGZ21} frameworks, indexing techniques~\cite{GubichevBS13,FletcherPP16,KuijpersFLY21,LiuWLLW21}, multi-query optimization~\cite{Abul-Basher17}, approximation~\cite{WadhwaPRBB19}, just-in-time compilation~\cite{TetzelLK20}, etc., in order to efficiently evaluate RPQs. However, these works have mainly focused on improving efficiency in terms of time, but not space.

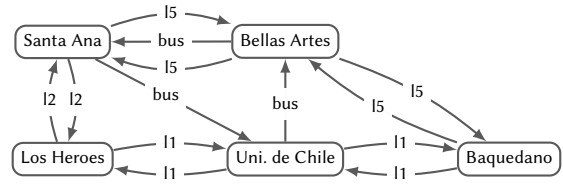
\begin{figure}[tb]
\setlength{\hgap}{1.5cm}
\setlength{\vgap}{1.1cm}
\centering
\begin{tikzpicture}
%
%
%
%
%

\newcommand{\bend}{10}
\newcommand{\bbend}{18}

\node[iri,anchor=mid] (ba) {Baquedano};

\node[iri,anchor=mid,left=\hgap of ba] (uc) {Uni.\ de Chile}
  edge[arrin,pos=0.5,bend right=\bend] node[lab] {l1} (ba)
  edge[arrout,pos=0.5,bend left=\bend] node[lab] {l1} (ba);
  
\node[iri,anchor=mid,left=\hgap of uc] (lh) {Los Heroes}
  edge[arrin,pos=0.5,bend right=\bend] node[lab] {l1} (uc)
  edge[arrout,pos=0.5,bend left=\bend] node[lab] {l1} (uc);
  
\node[iri,anchor=mid,above=\vgap of lh] (ana) {Santa Ana}
  edge[arrin,pos=0.5,bend right=\bbend] node[lab] {l2} (lh)
  edge[arrout,pos=0.5,bend left=\bbend] node[lab] {l2} (lh)
  edge[arrout] node[lab] {bus} (uc);
  
\node[iri,anchor=mid,above=\vgap of uc] (art) {Bellas Artes}
  edge[arrin,pos=0.5,bend right=\bend] node[lab] {l5} (ba)
  edge[arrout,pos=0.5,bend left=\bend] node[lab] {l5} (ba)
  edge[arrin,pos=0.5,bend right=\bbend] node[lab] {l5} (ana)
  edge[arrout,pos=0.5,bend left=\bbend] node[lab] {l5} (ana)
  edge[arrout,pos=0.5] node[lab] {bus} (ana)
  edge[arrin] node[lab] {bus} (uc);

\end{tikzpicture}
\caption{Santiago metro stations with metro lines and buses \label{fig:dg2}}
\end{figure}

\paragraph{Our Contribution}

We introduce a new technique to handle RPQs (with inverses) on labeled graphs that works on
a recent compressed representation of the graph called a ring~\cite{AHNRRS21}.
The ring was introduced for handling multijoin queries in worst-case optimal
time while using essentially the same space as a plain representation of all
the triples $(s,p,o)$ denoting edges $s \xrightarrow{p} o$ of the labeled graph. The ring represents those nodes
and edge labels as a sequence in a convenient order called the Burrows-Wheeler
Transform (BWT)~\cite{BW94} of the triples, and in turn represents the sequence using a wavelet tree data structure~\cite{GGV03}. Our technique combines (1) the 
backward search capabilities of the BWT, (2) the ability of the wavelet trees 
to efficiently work on ranges of nodes or edge labels, and (3) the regularity of the Glushkov automaton \cite{Glu61} of the regular expression and the versatility of its
bit-parallel simulation \cite{NR04}. As a result, we efficiently navigate towards the nodes of
the product graph that are involved in the solution of the RPQ, being able to process several NFA states and several graph nodes/labels simultaneously. Theorem~\ref{thm:main} shows that we spend logarithmic time per node and edge of the product subgraph induced by the query. Our experimental results show that we obtain a
time-competitive query solution within about twice the space of a compact data representation (because we need to duplicate all edges to handle reversed edges in the RPQs). This is still 3--5 times less than the space used by prominent alternative indexes that handle RPQs, while our index outperforms them in general: it was the fastest index on average in our benchmark, 1.67 times faster than Blazegraph, the second fastest index.

\section{Related Work}
\label{sec:related}


We now present related work on efficiently evaluating RPQs and related types of expressions, such as property paths in SPARQL.

\paragraph{Query planning} Various techniques have been proposed for evaluating path-based queries. Earlier works focused on evaluating shortest paths, such as the works on spatial networks~\citet{MiaoST07}, or reducing navigation to joins~\cite{GubichevN11}. Later works began to focus on RPQs, coinciding with the standardisation of SPARQL~1.1. 

Regarding navigation-based approaches, \citet{KoschmiederL12} propose to split an RPQ by its \textit{rare labels}, i.e., labels with fewer than $m$ edges where at least one such edge must be used in each path matching the RPQ; for example, given an expression $\rcat{\rcat{\rkle{a}}{b}}{\rkle{c}}$ (where $\rkle{\cdot}$ denotes Kleene star and $\rcat{\cdot}{\cdot}$ concatenation), if $b$ has few edges, the expression may be split into $\rcat{\rkle{a}}{b}$ and $\rcat{b}{\rkle{c}}$ to ensure more selective start/end points, and later joined. \citet{NoleS16} evaluate RPQs using the concept of \textit{Brzozowski derivatives}, whereby the regular expression is rewritten based on the symbols already read such that the rewritten expression matches suffixes that complete the path; for example, if the original expression is $\rcat{\rcat{\rkle{a}}{b}}{\rkle{c}}$, and we advance along an edge labeled $b$, then the derivative is $\rkle{c}$. \citet{WangWZ16} evaluate RPQs based on \textit{partial answers} that can be connected, allowing for these answers to not only be prefixes, but also infixes and suffixes; for example, if the original expression is $\rcat{\rcat{\rkle{a}}{b}}{\rkle{c}}$, partial answers corresponding to $\rkle{a}$ (prefix), $\rcat{b}{c}$ (infix) and $\rkle{c}$ (suffix) can be generated in parallel and combined. \citet{NguyenK17} split RPQs similarly to the ``rare labels'' strategy, but rather minimize the cost of the most costly sub-RPQ resulting from the split. \citet{WadhwaPRBB19} compute approximate RPQ results using bidirectional random walks, where a forward walk begins from the source node, a backward walk begins from the target node, and walks that ``meet'' are reported as solutions.

Other approaches evaluate RPQs using recursive (query) languages. \citet{DeyC0GWL13} evaluate RPQs using Datalog or recursive SQL queries; they further return \textit{provenance} in the form of all edges involved in some or all matching paths. \citet{YakovetsGG13} likewise propose to translate property paths into recursive SQL queries, but note that the resulting queries can be complex and difficult to optimize. \citet{JachietGGL20} propose an extended relational algebra with a transitivity/fixpoint operator, and describe how RPQs (more specifically, unions of conjunctive RPQs) can be translated to this algebra. \citet{FiondaPC19} propose \textit{extended property paths}, which includes difference and intersection over paths, as well as the ability to express tests that constrain nodes along the path; non-recursive expressions are translated into SPARQL~1.1, while recursive expressions require a recursive extension of SPARQL.

Combining both navigational/automata and recursive/relational approaches, \citet{YakovetsGG16} propose hybrid ``waveplans'' that can mix operators from both algebras and thus can express novel query plans. \citet{Abul-Basher17} propose a related framework called ``swarmguide'' for optimizing multiple RPQs at once, based on finding a maximum common sub-automaton for the RPQs, which can be evaluated and reused across RPQs using views.

Finally, a number of approaches leverage software or hardware acceleration techniques. \citet{MiuraAK19} evaluate RPQs on top of field programmable gate arrays (FPGAs), which enable high degrees of parallelism; specifically, the RPQ is split into multiple ``stages'', where sort--merge joins are applied on the FPGA to join results from different stages in a pipelined manner. \citet{TetzelLK20} use just-in-time compilation techniques in order to generate native C++ code that directly evaluates the RPQ on the graph.

\paragraph{Indexing} Custom indexing approaches have also been proposed for optimizing RPQ evaluation. \citet{GubichevBS13} extend RDF-3X with support for property paths using an indexing technique called FERRARI~\cite{SeufertABW13}, based on encoding the transitive closure of the graph induced by a given edge label using (potentially overapproximated) intervals of node ids.  \citet{WangRJLYF14} propose a predicate-based indexing scheme to evaluate RPQs over RDF graphs, where four orders are indexed -- \textsc{pso}, \textsc{pos}, \textsc{ps}, \textsc{po} -- in order to efficiently evaluate triple patterns with a fixed predicate. \citet{FletcherPP16} propose a \textit{$k$-path index} that indexes all paths of length up to $k$ in a B$^+$-tree, specifically indexing the word of the path, the source node, and the target node. \citet{KuijpersFLY21} describe the use of $k$-path indexes to optimize the evaluation of Cypher queries in the Neo4j graph database, while \citet{LiuWLLW21} use similar indexes, which they populate with frequent paths mined from the graph.

\paragraph{Other settings} We focus on evaluating RPQs over a static graph on a single machine. However, we briefly mention some other works on evaluating RPQs in other settings. A number of works have looked into enabling horizontal scale by evaluating RPQs over RDF graphs distributed/partitioned over multiple machines~\cite{NoleS16,WangWZ16,ColazzoMNS18,MehmoodSSNd19,GuoGZ21}, sometimes using existing frameworks such as Pregel~\cite{NoleS16} or MapReduce~\cite{ColazzoMNS18}. Other works have explored the evaluation of property paths/RPQs over Linked Data in the decentralized setting, whereby RDF graphs on the Web are navigated dynamically while evaluating the RPQ~\cite{BaierDRV17,HartigP17}. A recent work has explored the evaluation of RPQs over sliding windows of streaming graph data~\cite{PacaciBO20}.

\paragraph{Novelty} We introduce a novel technique to evaluate RPQs (with inverses) that is efficient both in time and space. While some indexing schemes explore a time-space trade-off, they occupy space \textit{additional} to representing and indexing the graph~\cite{SeufertABW13}. To the best of our knowledge, our approach is the first that can efficiently evaluate RPQs on a compressed representation of the graph, and the first to evaluate RPQs based on Glushkov automata \cite{Glu61}, highlighting key advantages of this construction: It not only enables a more space-efficient bit-parallel simulation of the NFA \cite{NR04}, but also its transitions exhibit a regularity that is crucial to efficiently evaluating RPQs. The combination of the backward search capabilities of the BWT \cite{BW94}, the ability of the wavelet trees \cite{GGV03} to work on ranges of nodes/labels, and the regularity of Glushvov's automaton, allow us to simulate the traversal of {\em only} the product subgraph induced by the RPQ (without spending time on outgoing edges). The bit-parallel simulation and the ability to work on ranges of nodes and labels further enable processing sets of nodes of the product graph {\em simultaneously}, thus speeding up the classical evaluation strategy.

\section{Basic Concepts}


\subsection{Regular Path Queries}

Let $\C$ denote a set of symbols. We define a (directed edge-labeled) graph $G \subseteq \C \times \C \times \C$ to be a finite set of triples of symbols of the form $(s,p,o)$, denoting (subject,predicate,object). Each triple of $G$ can be viewed as a labeled edge of the form $s \xrightarrow{p} o$. Given a graph $G$, we define the \emph{nodes} of $G$ as $V = \{ x \mid \exists\, y,z,~(x,y,z) \in G \lor (z,y,x)\in G\}$.


A \emph{path} $\rho$ from $x_0$ to $x_n$ in a graph $G$ is a string of the form $x_0\,p_1\,x_1\,\ldots\,p_n\,x_{n}$ such that $(x_{i-1},p_i,x_i) \in G$ for $1 \leq i \leq n$. Abusing notation, we may write that $\rho \in G$ if $\rho$ is a path in $G$. 
We call $\lab{\rho} = p_1 \ldots p_n \in \C^*$ the \emph{word} of $\rho$.


We say that $\varepsilon$ is a \emph{regular expression}, and that any element of $\C$ is a regular expression. If $E, E_1$ and $E_2$ are regular expressions, we say that $\rkle{E}$ (Kleene closure), $\rcat{E_1}{E_2}$ (concatenation) and $\ror{E_1}{E_2}$ (disjunction) are also regular expressions. We may further use $\rplu{E}$ as an abbreviation for $\rcat{\rkle{E}}{E}$, and $\ropt{E}$ as an abbreviation for $\ror{\epsilon}{E}$.

We define by $\Ci = \{ \inv{s} \mid s \in \C \}$ the \emph{inverses} of the symbols of $\C$, and by $\Cc = \C \cup \Ci$ the set of symbols and their inverses. We assume that $\C \cap \Ci = \emptyset$ and that $s = \inv{(\inv{s})}$. We denote by $\inv{G} = \{ (y,\inv{p},x) \mid (x,p,y) \in G \}$ the \emph{inverse} of a graph $G$, and by $\comp{G} = G \cup \inv{G}$ the \emph{completion} of $G$. If $E$ is a two-way regular expression, then so is $\rinv{E}$ (inverse). A two-way regular expression on $\C$ can be rewritten to a regular expression on $\Cc$.


Given a regular expression $E$, we denote by $L(E)$ the language of $E$, and we say that a path $\rho$ \emph{matches} $E$ if and only if $\lab{\rho} \in L(E)$.

Let $\V$ denote a set of variables. 
Let $\mu : \V \rightarrow \C$ denote a partial mapping from variables to symbols. We denote the domain of $\mu$ as $\dom{\mu}$, which is the set of variables for which $\mu$ is defined. 
If $E$ is a regular expression, $s \in \V \cup \C$ and $o \in \V \cup \C$, then we call $(s,E,o)$ a \emph{regular path query} (\emph{RPQ}). Let $x_\mu$ be defined as $\mu(x)$ if $x \in \dom{\mu}$, or $x$ otherwise. We define the \emph{evaluation} of $(s,E,o)$ on $G$ as:
\begin{align*}
(s,E,o)(G) = \{ \mu \mid ~& \dom{\mu} = \{s,o \} \cap \V \text{ and there exists a path }\rho \\[-0.6ex]
 & \text{ from }s_\mu\text{ to }o_\mu\text{ in }G \text{ such that }\rho\text{ matches }E\}.
\end{align*}
%
%

\vspace{-4mm}
\begin{example}
Take the graph $G$ of Fig.~\ref{fig:dg2} and the RPQ $(x,\metro,y)$, where $x,y \in \V$ are variables. Infinitely many paths in $G$ match the expression $\metro$, including (abbreviating node labels):
\begin{center}
\gf{UCh} $~~$\gf{l1} $~~$\gf{LH} $~~$\gf{l1} $~~$\gf{UCh} \\
\gf{UCh} $~~$\gf{l1} $~~$\gf{LH} $~~$\gf{l1} $~~$\textsf{UCh} $~~$\gf{l1} $~~$\gf{LH} \\
\gf{Baq} $~~$\gf{l1} $~~$\gf{UCh} $~~$\gf{l1} $~~$\gf{LH} $~~$\gf{l2} $~~$\gf{SA} 
\end{center}
\noindent and so forth. The evaluation of the RPQ on $G$ will return all mappings such that $x$ maps to the start node of some such path, and $y$ maps to the end node of the same path. For example, from the first path, we will return a solution $\mu$ such that $\mu(x) = \gf{UCh}$, $\mu(y) = \gf{UCh}$, and $\mu$ is undefined for all other variables. The evaluation is finite as it can map $x$ and $y$, at most, to all pairs of nodes in $G$. 

If we instead consider $(\gf{Baq},\metro,y)$, where $\gf{Baq} \in \C$ and $y \in \V$, then its evaluation on $G$ will return all mappings such that $y$ maps to the end node of a path that starts with $\gf{Baq}$. For example, from the third path listed previously, we would return $\mu$ such that $\mu(y) = \gf{SA}$ and $\mu$ is undefined for all other variables.

Finally, the evaluation of $(\gf{Baq},\metro,\gf{SA})$, where $\gf{Baq},\gf{SA} \in \C$, returns a single solution $\mu$ that is undefined for all variables; if \gf{SA} were not reachable from \gf{Baq} via a path matching $\metro$, then no solution would be returned.
\end{example}

If $E$ is a two-way regular expression over $\C$, $s \in V \cup \V$ and $o \in V \cup \V$, we call $(s,E,o)$ a \emph{two-way regular path query} (2RPQ). We define the evaluation of the 2RPQ $(s,E,o)$ on $G$ as the evaluation of the RPQ $(s,E',o)$ on $\comp{G}$, where $E'$ is the rewritten form of $E$ using only atomic inverses, and is thus a regular expression over $\Cc$.


\subsection{Product Graph}

One approach for evaluating an RPQ $(s,E,o)$ on $G$ involves computing the \emph{product graph} of $G$~\cite{MendelzonW95}. Specifically, we first convert the regular expression $E$ into a non-deterministic finite automaton (NFA) $M_E = (Q, \Sigma_E, \Delta, q_0, F)$, where $Q$ denotes the set of states, $\C_E \subseteq \C$ the set of symbols used in $E$, $\Delta$ the transitions, $q_0$ the initial state, and $F$ the set of accepting states. The conversion from a regular expression to an NFA can be conducted  (for example) using Thompson's classical algorithm, where we assume that $\varepsilon$-transitions have been (subsequently) removed from $M_E$. Letting $V$ denote the nodes of $G$, then the \emph{product graph} $G_E \subseteq (V \times Q) \times (V \times Q)$ of $G$ with respect to $E$ is a directed unlabeled graph defined as follows:
\[
G_E = \{ ((x,q_x),(y,q_y)) \mid \exists\, p \in \C_E, 
 (x,p,y) \in G \land (q_x,p,q_y) \in \Delta\}.
\]
\noindent
The RPQ can then be evaluated by using standard graph search algorithms (e.g., BFS, DFS, etc.) to find paths in the product graph $G_E$ that start from some node $(x,q_0) \in 
V \times \{ q_0 \}$ and end in some node $(y,q_f) \in V \times F$ (such that $x = s$ if $s \in \Sigma$, and $y = o$ if $o \in \Sigma$).

\subsection{Bit-parallel Glushkov Automata}

Consider a regular expression $E$ on alphabet $\C$ with $m$ occurrences of symbols in $\Sigma$. Compared to the 
classical Thompson's construction of an NFA from $E$, Glushkov's 
\cite{Glu61,BS86} has the disadvantage of generating $\Theta(m^2)$ edges in 
the worst case, and needing $O(m^2)$ construction time \cite{BK93}. In exchange,
it has various properties that will make it interesting for our purposes:

\begin{enumerate}
\item The NFA has no $\varepsilon$-transitions.
\item The NFA has exactly $m+1$ states, optimal in the worst case.
\item All the transitions arriving at a state have the same label.
\end{enumerate}

These properties imply the following important fact.

\begin{fact} \label{fact:1}
In a Glushov NFA, the states reached in one step from a set $X$ of states by symbol $c$ are the intersection of those reached from $X$ in one step and 
those states reached by $c$ from any state.
\end{fact}

\begin{figure}[t]
    \centering
    \includegraphics[width=0.25\textwidth]{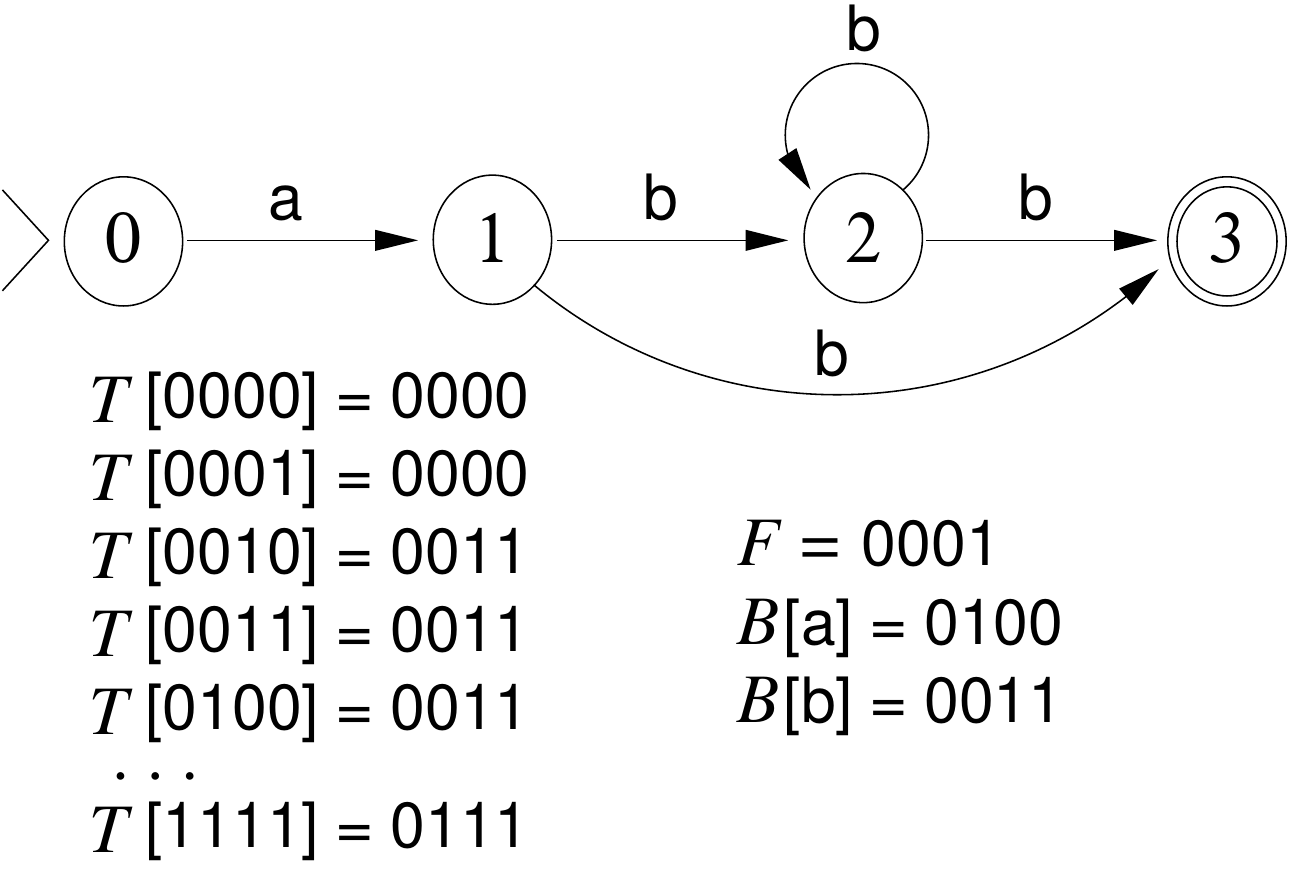}
    \caption{The Gluskov automaton of the regular expression $\rcat{\rcat{a}{\rkle{b}}}{b}$, and its bit-parallel representation below it.}
    \label{fig:glushkov}
\end{figure}

\begin{example}
The top of Fig.~\ref{fig:glushkov} gives the Gluskov automaton for $\rcat{\rcat{a}{\rkle{b}}}{b}$. Take the states $X = \{0,2\}$. The states $\{2,3\}$ are reachable from $X$ in one step via $b$, which is equal to the intersection of the states $\{1,2,3\}$ reachable in one step from $X$ via any symbol and the states $\{ 2,3 \}$ reachable in one step via $b$ from any state.
\end{example}

This property enables the \emph{bit-parallel} simulation of the NFA
\cite{NR04}. This simulation represents NFA states as bits in a computer 
word, so each configuration of active and inactive states (bits set to 1 and 0,
respectively), correspond to a state in the DFA according to the classic 
powerset construction. The simulation operates on all the states in parallel
by using the classic arithmetical and logical operations on computer words.
Assume for simplicity that the bits of the NFA states fit in a single 
computer word; we discuss the general case later. Further assume that the alphabet is an integer range $\Sigma = [1\dd\sigma]$. The simulation maintains the following
variables:

\begin{itemize}
\item A computer word $D$ holding $m+1$ bits tells, at every step, the 
active NFA states, as discussed. Assume the initial state corresponds to the
highest bit.
\item A table $B[1\dd\sigma]$ of computer words indicates with 1s, at each 
$B[c]$, the NFA states that are targets of transitions labeled $c$.
\item A table $T[0\dd 2^{m+1}-1]$ stores in $T[X]$, for each possible 
$(m+1)$-bit argument $X$ representing a set of NFA states, the states 
reachable from states in $X$ in one step, by any symbol.
\item A computer word $F$ marks with 1s the final NFA states.
\end{itemize}

The simulation is then carried out as follows:

\begin{enumerate}
\item We set $D \leftarrow 2^m$ to activate the initial state.
\item If $D ~\&~ F ~\not=~ 0$, then we have reached a final state and accept the word read
(recall that `\&' is the bitwise-and).
\item If $D=0$, then we have run out of active states and reject.
\item For each input symbol $c$, we use Fact~\ref{fact:1} to update $D$ as follows, such that the new active states are those that are reached from the current ones and also reached by symbol $c$:
\begin{equation} \label{eq}
 D ~~\gets~~ T[D] ~\&~ B[c], 
\end{equation}
\item Return to point 2.
\end{enumerate}

\begin{example}
Fig.~\ref{fig:glushkov} shows the Glushkov automaton of $\rcat{\rcat{a}{\rkle{b}}}{b}$ and its bit-parallel representation.
Given a string $S= abba$, we initialize $D \leftarrow 1000$ with the initial state $0$ activated. We now read $S[1] = a$ and update $D \gets T[1000] ~\&~ B[a] = 0100 ~\&~ 0100 = 0100$, activating state $1$. We read $S[2]=b$ and update $D \gets T[0100] ~\&~ B[b] = 0011 ~\&~ 0011 = 0011$, indicating that states $2$ and $3$ are now active. We report here the endpoint of a match since $D ~\&~ F = 0011 ~\&~ 0001 \not= 0000$. To find other endpoints, we next read $S[3] = b$ and update $D \gets T[0011] ~\&~ B[b] = 0011 ~\&~ 0011 = 0011$, reporting this position as well. Finally, we read $S[4]=a$ and update $D \gets T[0011] ~\&~ B[a] = 0011 ~\&~ 0100 = 0000$, so we run out of active states and finish.
\end{example}

The space of the simulation is $O(2^m+\sigma)$, instead of the
worst-case $O(2^m \sigma)$ of a classical DFA implementation. The tables are
built in time $O(2^m)$ by using lazy initialization for $B$.

A similar simulation can be used to read the text
in reverse order~\cite{NR04} by building a table $T'[0\dd 2^m-1]$ where $T'[X]$ marks with
1s the states that can reach some state in $X$ in one step, initializing $D
\gets F$ and, for each new symbol $c$, updating
\begin{equation} \label{eq:revglus}
 D ~~\gets~~ T'[D ~\&~ B[c]], 
\end{equation}
and accepting when $D ~\&~ 2^m \not= 0$.

Bit-parallelism uses the RAM model of computation, where all the arithmetical
and logical operations over a $w$-bit word take constant time; it is usual to
assume $w = \Theta(\log n)$, where $n$ is the data size. In our case, if 
$m+1 > w$, then we need to use $\lceil (m+1)/w \rceil$ computer words to hold
$D$, $F$, and every entry of $B$ and $T$. In this case, all the time and space
complexities get multiplied by $O(m/w)$. Furthermore, if we want to avoid the 
exponential space and time $O(2^m)$, we can split table $T$ vertically into $d$-bit 
subtables $T_1, \ldots, T_{\lceil (m+1)/d\rceil}$, so that if we partition 
$X = X_1 \cdots X_{\lceil (m+1)/d\rceil}$, then $T[X] = T_1[X_1] ~| \cdots |~ 
T_{\lceil (m+1)/d\rceil}[X_{\lceil (m+1)/d\rceil}]$, where ``$|$'' denotes the 
bitwise-or. This reduces the space to $O((m/d) 2^d + \sigma)$ and multiplies
the time by $O(m/d)$ instead of $O(m/w)$, for any desired $1\le d\le 
\min(w,m+1)$ \cite{NR04}. We will assume for simplicity that $m=O(w)$ and
use $O(2^m)$ space throughout the paper.

\subsection{The Ring}

The {\em ring} \cite{AHNRRS21} is a novel representation for a set of triples $(s,p,o)$, which supports worst-case optimal multijoin queries using the Leapfrog Triejoin algorithm \cite{Vel14}. The ring regards the $n$ triples as a set of $n$ circular strings $spo$ (or $pos$, or $osp$) of length $3$. It then
    creates three strings by shifting and sorting the circular strings:
    \begin{itemize}
        \item $L_\mathrm{o}[1\dd n]$ lists the objects $o$ that (circularly) precede the lexicographically sorted strings $spo$.
        \item $L_\mathrm{s}[1\dd n]$ lists the subjects $s$ that (circularly) precede the lexicographically sorted strings $pos$.
        \item $L_\mathrm{p}[1\dd n]$ lists the predicates $p$ that (circularly) precede the lexicographically sorted strings $osp$.
    \end{itemize}
    The concatenation $L_\mathrm{o} \cdot L_\mathrm{s} \cdot L_\mathrm{p}$ is indeed the Burrows-Wheeler Transform (BWT) \cite{BW94} of the concatenation of all the triples 
    (with some tweaks, see the original article \cite{AHNRRS21} for details).

With this arrangement, a range in $L_\mathrm{o}$ corresponds to a lexicographic interval of triples $spo$. In particular, a range may represent all the triples with a specific subject $s$ (i.e., strings starting with $s$), and a smaller range may represent all the triples with  subject $s$ and predicate $p$ (i.e., strings starting with $sp$). A range in $L_\mathrm{o}$ can also represent a range of subjects $s_b\dd s_e$, and even a subject $s$ followed by a range of predicates $p_s \dd p_e$. Analogously, ranges in $L_\mathrm{s}$ correspond to lexicographic intervals of triples $pos$ and ranges in $L_\mathrm{p}$ correspond to lexicographic intervals of triples $osp$. Note that, in the three strings, the range $[1\dd n]$ represents all the triples and a range of size $1$ represents an individual triple.

The ring retrieves triples using so-called {\em LF-steps}, defined on array $L_\mathrm{p}$ (and analogously on $L_\mathrm{s}$ and $L_\mathrm{o}$), as follows:
\begin{equation} \label{eq:lf}
\LF_\mathrm{p}(i) ~~=~~ C_\mathrm{p}[c] + \rank_c(L_\mathrm{p},i),
\end{equation}
where $c=L_\mathrm{p}[i]$, $C_\mathrm{p}[c]$ counts the symbols smaller than $c$ in $L_\mathrm{p}$, and $\rank_c(L_\mathrm{p},i)$
counts the number of times $c$ occurs in $L_\mathrm{p}[1\dd i]$. It then holds that the subject of the triple represented at $L_\mathrm{p}[i]$ is $L_\mathrm{s}[i']$ for $i'=\LF_\mathrm{p}(i)$, and the object is $L_\mathrm{o}[i'']$ for $i'' = \LF_\mathrm{s}(i')$. It further holds that $i = \LF_\mathrm{o}(i'')$, where we find the predicate at $L_\mathrm{p}[i]$.

\begin{figure}[t]
    \centering
    \includegraphics[width=0.45\textwidth]{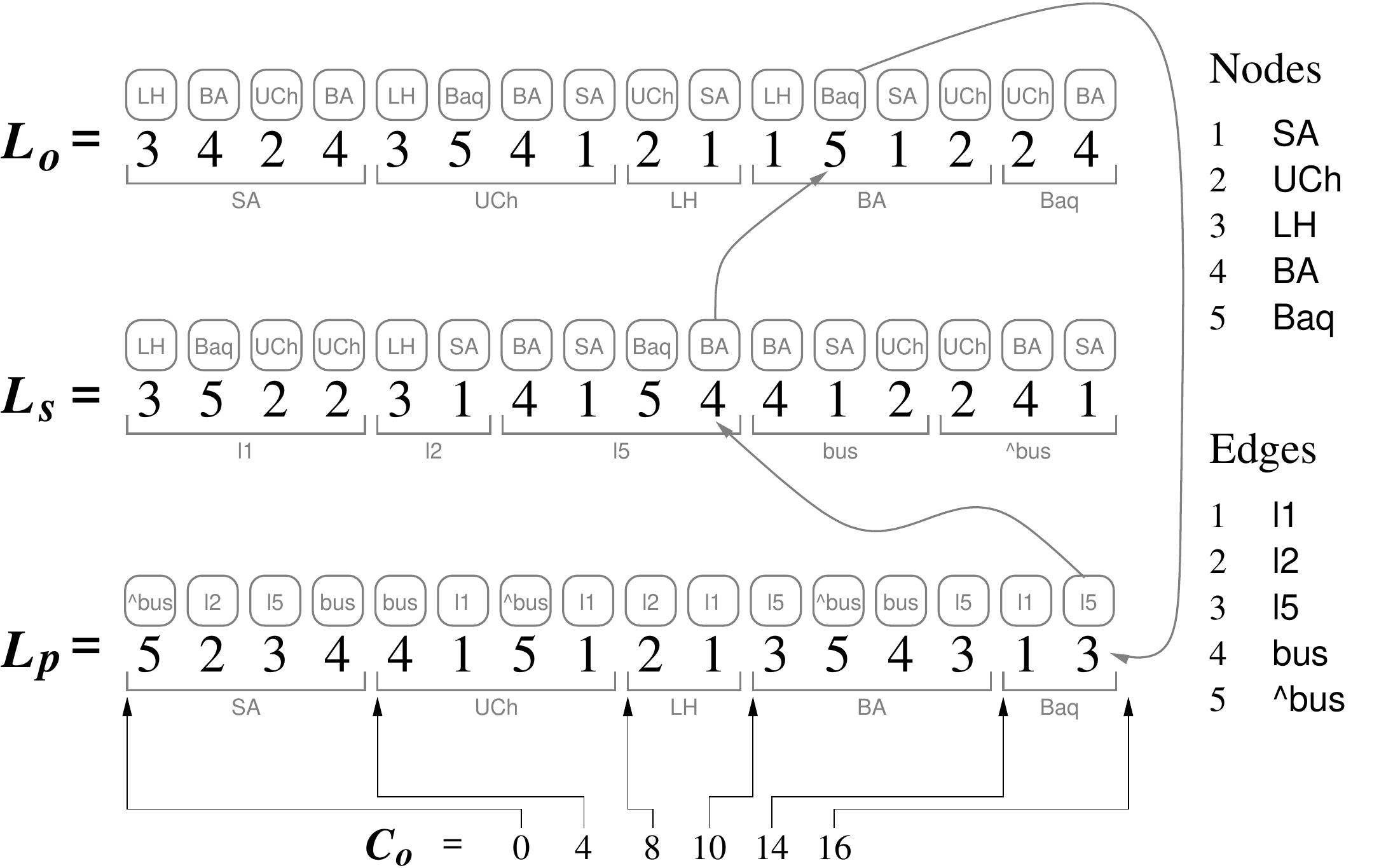}
    \caption{The ring data structure built on the graph of Fig.~\ref{fig:dg2}, which we have completed by adding a reverse edge labeled {\rm \textsf{\^{}bus}} for each edge labeled {\rm \textsf{bus}} ({\rm\textsf{l1}}, {\rm\textsf{l2}} and {\rm\textsf{l5}} are bidirectional). We also show how the last triple in $L_\mathrm{p}$ is tracked.}
    \label{fig:bwt}
\end{figure}

\begin{example}
Fig.~\ref{fig:bwt} shows the ring for the completion of the graph of Fig.~\ref{fig:dg2}, adding reversed edges labeled \textsf{\^{}bus}. On the right we map the (abbreviated) nodes and edge labels to integers. We write the abbreviated names over the numbers in the sequences $L_\mathrm{o}$, $L_\mathrm{s}$, and $L_\mathrm{p}$ for readability. Note that, for example, $L_\mathrm{p}$ can be partitioned into the triples $osp$ starting with objects $1$ (\textsf{SA}), $2$ (\textsf{UCh}), $3$ (\textsf{LH}), $4$ (\textsf{BA}), and $5$ (\textsf{Baq}), which we indicate below the sequence, and whose endpoints are marked in the array $C_\mathrm{o}$, shown on the bottom.

Consider the triple referenced from $L_\mathrm{p}[16]$. It refers to the object $5$ (\textsf{Baq}) because it belongs to the range $L_\mathrm{p}[15\dd 16]=L_\mathrm{p}[C_\mathrm{o}[5]+1\dd C_\mathrm{o}[5+1]]$. The value $L_\mathrm{p}[16]=3$ (\textsf{l5}) gives the label of the edge. To find the corresponding subject, we note that this is the fourth $3$ (\textsf{l5}) in $L_\mathrm{p}$. Then, if we go to the fourth position in the area of \textsf{l5} in $L_\mathrm{s}$, $L_\mathrm{s}[7\dd 10]$, which is $L_\mathrm{s}[10]$, we learn that the subject is $L_\mathrm{s}[10]=4$ (\textsf{BA}). Indeed, $\LF_\mathrm{p}(16)=10$. Thus, the triple is \textsf{BA} $\stackrel{\!\mathsf{l5}}{\longrightarrow}$ \textsf{Baq}. Further, $L_\mathrm{s}[10]$ is the second $4$ in $L_\mathrm{s}$, so if we go to the corresponding position $L_\mathrm{o}[12]$ (note $\LF_\mathrm{s}(10)=12$) we cyclically find $L_\mathrm{o}[12]=5$ (\textsf{Baq}), the object of the triple. We indeed return to position $L_\mathrm{p}[16]$ if we map $L_\mathrm{o}[12]$, the second $5$ in $L_\mathrm{o}$, to $L_\mathrm{p}$. Again, $\LF_\mathrm{o}(12)=16$. 
\end{example}

The key to solving multijoins with the ring is the so-called {\em backward search}, which computes in batch all the LF-steps in a range. Consider a range $L_\mathrm{p}[b_o\dd e_o]$ listing, say, all the triples with a specific object $o$ (i.e., all the triples $osp$ for any $s$ and $p$). The backward search by some specific predicate $p$ gives the range $L_\mathrm{s}[b_p\dd e_p]$ corresponding to all the triples with object $o$ and predicate $p$ (i.e., all the triples $pos$ for any $s$). This is computed with the following formula, which extends the LF-steps (Eq.~(\ref{eq:lf})) to ranges
\cite{FM05,AHNRRS21}:
\begin{eqnarray}
b_p &=& C_\mathrm{p}[p] + \rank_p(L_\mathrm{p},b_o-1)+1, \label{eq:bwdsearch} \\
e_p &=& C_\mathrm{p}[p] + \rank_p(L_\mathrm{p},e_o). \label{eq:bwdsearch2}
\end{eqnarray}
Listing the subjects $s$ in $L_\mathrm{s}[b_p\dd e_p]$ then yields all the triples with that specific predicate $p$ and object $o$, for example.

\begin{example}
Continuing our example, let us start from $L_\mathrm{p}[11\dd 14]$, corresponding to object \textsf{BA}. If we apply a backward search step from $b_o=11$ and $e_o=14$, on the label $3$ (\textsf{l5}) using Eqs.~(\ref{eq:bwdsearch}) and (\ref{eq:bwdsearch2}), we obtain $L_\mathrm{s}[b_s\dd e_s]=L_\mathrm{s}[8..9]=\langle 1,5 \rangle$, showing that we arrive at \textsf{BA} by \textsf{l5} from sources $L_\mathrm{s}[8]=1$ (\textsf{SA}) and $L_\mathrm{s}[9]=5$ (\textsf{Baq}).
\end{example}

The ring uses a data structure called a {\em wavelet tree} \cite{GGV03} to represent each of the sequences $L_\mathrm{o}$, $L_\mathrm{s}$, and $L_\mathrm{p}$, enabling the efficient evaluation of queries like $\rank_p(L_\mathrm{p},i)$. 

\subsection{Wavelet trees}
\label{sec:wtrees}

The wavelet tree represents a string $L[1\dd n]$ over an alphabet $[1\dd \sigma]$ as a perfect binary tree with $\sigma$ leaves, one per symbol, so that the $c$\textsuperscript{th} left-to-right leaf represents symbol $c$. Each internal wavelet tree node $v$ that is the ancestor of leaves $c_s\dd c_e$ represents the subsequence $L_{\langle c_s,c_e\rangle}$ of $L$ formed by the symbols in $c_s \dd c_e$. Instead of storing $L_{\langle c_s,c_e\rangle}$, node $v$ stores a bitvector $W_{\langle c_s,c_e\rangle}$, so that $W_{\langle c_s,c_e\rangle}[i]=0$ iff the leaf representing symbol $S_{\langle c_s,c_e\rangle}[i]$ descends by the left child of $v$. The leaves are conceptual and not stored. It is not hard to see that all the bitvectors stored at the internal wavelet tree nodes amount to $n\log\sigma$ bits, the same as a plain representation of $L$ (our logarithms default to base $2$).

The wavelet tree obtains $L[i]$ in $O(\log\sigma)$ time as follows. Let $v$ be the wavelet tree root, thus it stores bitvector $W = W_{\langle 1,\sigma\rangle}$ where $W[i]=0$ indicates that $L[i] \in [1\dd \sigma/2]$; otherwise $L[i] \in [\sigma/2+1\dd \sigma]$ (we assume that $\sigma$ is a power of $2$ for simplicity of presentation). In the first case, $L[i] = L_{\langle 1,\sigma\rangle}[i]$ corresponds to $L_{\langle 1,\sigma/2\rangle}[i']$, where $i' = \rank_0(W,i)$ and we continue recursively by the left child of $v$ with position $i'$. In the second case, $L[i]$ corresponds to $L_{\langle \sigma/2+1,\sigma\rangle}[i'']$, where $i'' = \rank_1(W,i)$ and we continue recursively by the right child of $v$ with position $i''$.

Operation $\rank$ on bitvectors can be done in $O(1)$ time by adding only sublinear space on top of the bitvector \cite{Cla96,Mun96}. Therefore, in time $O(\log\sigma)$ we arrive at a leaf and determine $L[i]$. The total space of the wavelet tree is $n\log\sigma + o(n\log\sigma)+O(\sigma \log n)$ bits, the latter term for the tree pointers. Note that $O(\sigma\log n)$ also absorbs the space of the arrays $C_x$ used for backward search.

A similar algorithm can be used to compute $\rank_c(L,i)$. We start at the wavelet tree root $v$ and, if $c$ descends by the left child, we recursively go left with $i \gets \rank_0(W,i)$; otherwise we recursively go right with $i \gets \rank_1(W,i)$. When we arrive at the leaf $c$, the current value of $i$ is the answer. Furthermore, the number of leaf positions to the left of $c$ is precisely $C[c]$, which directly gives the values of the LF and the backward search formulas (Eqs.~(\ref{eq:lf}) to (\ref{eq:bwdsearch2})).

\begin{figure}[t]
    \centering
    \includegraphics[width=0.48\textwidth]{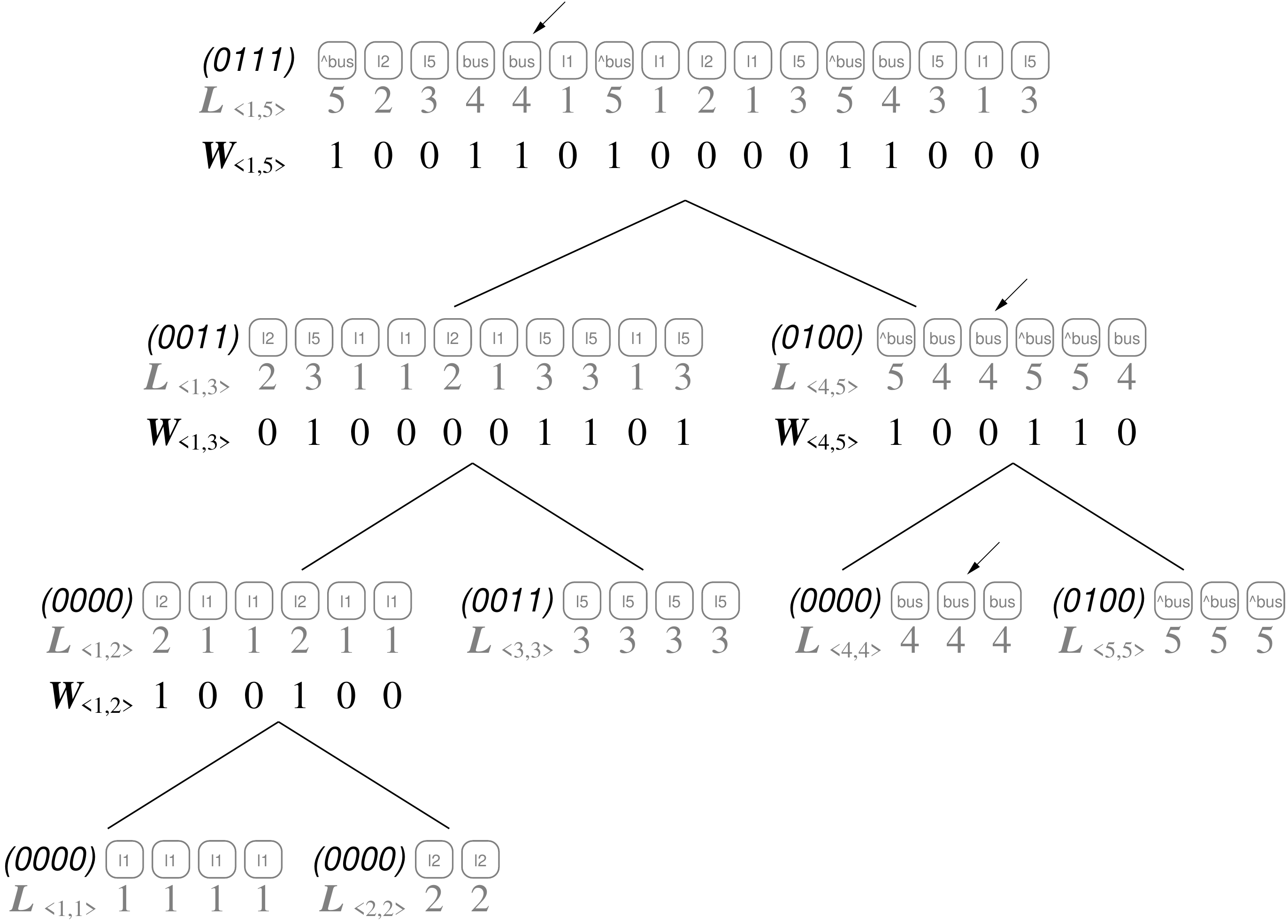}
    \caption{The wavelet tree of the sequence $L_\mathrm{p}$ of Fig.~\ref{fig:bwt}. The short diagonal arrows track $L[5]$. The slanted bitvectors on the nodes refer to the $B$ entries of the automaton of Fig.~\ref{fig:rpqglus}.}
    \label{fig:wt}
\end{figure}

\begin{example}
Fig.~\ref{fig:wt} shows the wavelet tree of sequence $L_\mathrm{p}$ for our running example (ignore the slanted bitvectors for now). To compute $\rank_4(L_\mathrm{p},5)$, we start at position $i\gets 4$ of the root (the short diagonal arrows track our position along the tree). Since leaf $4$ is to the right, we go right and set $i \leftarrow \rank_1(W_{\langle 1,5\rangle},5) = 3$. On the right child of the root, we see that leaf $4$ descends to the left, so we go left with $i \leftarrow \rank_0(W_{\langle 4,5\rangle},3)=2$, arriving at the leaf of $4$. Thus $\rank_4(L_\mathrm{p},5)=i=2$. The lengths of all the leaves to the left add up to $C_\mathrm{p}[4]=10$, so adding $i$ we obtain position $12=\LF_\mathrm{p}(5)$.
\end{example}

Wavelet trees can be used for many other purposes \cite{GNP12,Nav14}. We will indeed make use of their extended capabilities for our algorithm. A good warmup is the following algorithm to enumerate all distinct symbols in $L[b\dd e]$: We start at the root and descend to the left with the interval $L_{\langle 1,\sigma/2\rangle}[b'\dd e']$, where $b' = \rank_0(W,b-1)+1$ and $e'=\rank_0(W,e)$. We also descend to the right with the interval $L_{\langle \sigma/2+1,\sigma\rangle}[b''\dd e'']$, where $b'' = \rank_1(W,b-1)+1$ and $e''=\rank_1(W,e)$. We abandon every empty interval and instead report every leaf we arrive at (we later exemplify more complex variants of this algorithm). The total time is then $O(\log\sigma)$ per distinct symbol reported, irrespective of the total number of symbols.


\section{Our Approach}
\label{sec:our-approach}

We will use part of the ring's structure to navigate backwards all the paths that match a given 2RPQ. More precisely, we use the wavelet trees representing sequences $L_\mathrm{p}$ and $L_\mathrm{s}$, as well as all the arrays $C_*$.

The sets of subjects and objects are equal and correspond to the nodes $V$
in the graph; each node may act as a subject (i.e., edge source) or as an 
object (i.e., edge target). The set of predicates $P \subseteq \Cc$ corresponds to the edge labels of $\comp{G}$.

We will first focus on 2RPQs of the form $(x,E,o)$, where $x \in \V$ and $o \in V$. We will build the Glushkov automaton for $E$ and use it to navigate backwards, from objects towards subjects. Since we use the NFA backwards, we will start from its final states, $D=F$, use the reverse Glushkov simulation of Eq.~(\ref{eq:revglus}), and report a valid binding $x=s$ at every node $s \in V$ where the initial NFA state is activated. The navigation will start from the range of $o$ in $L_\mathrm{p}$.

This technique also handles 2RPQs of the form $(s,E,y)$, where $s \in V$ and $y \in \V$, by simply reversing $E$ and searching instead for $(y,\inv{E},s)$.
We will later consider the other kinds of 2RPQs.

We note
that since the alphabet of $E$ is $P$, our vector $B[1\dd|P|]$ for the bit-parallel NFA simulation is of size $O(|P|)$, but
still preprocessing the RPQ takes time $O(2^m)$ with lazy initialization. 
This adds a working space usage of $O(2^m + |P|)$ on top of the ring.

\begin{figure}
    \centering
    \includegraphics[width=0.25\textwidth]{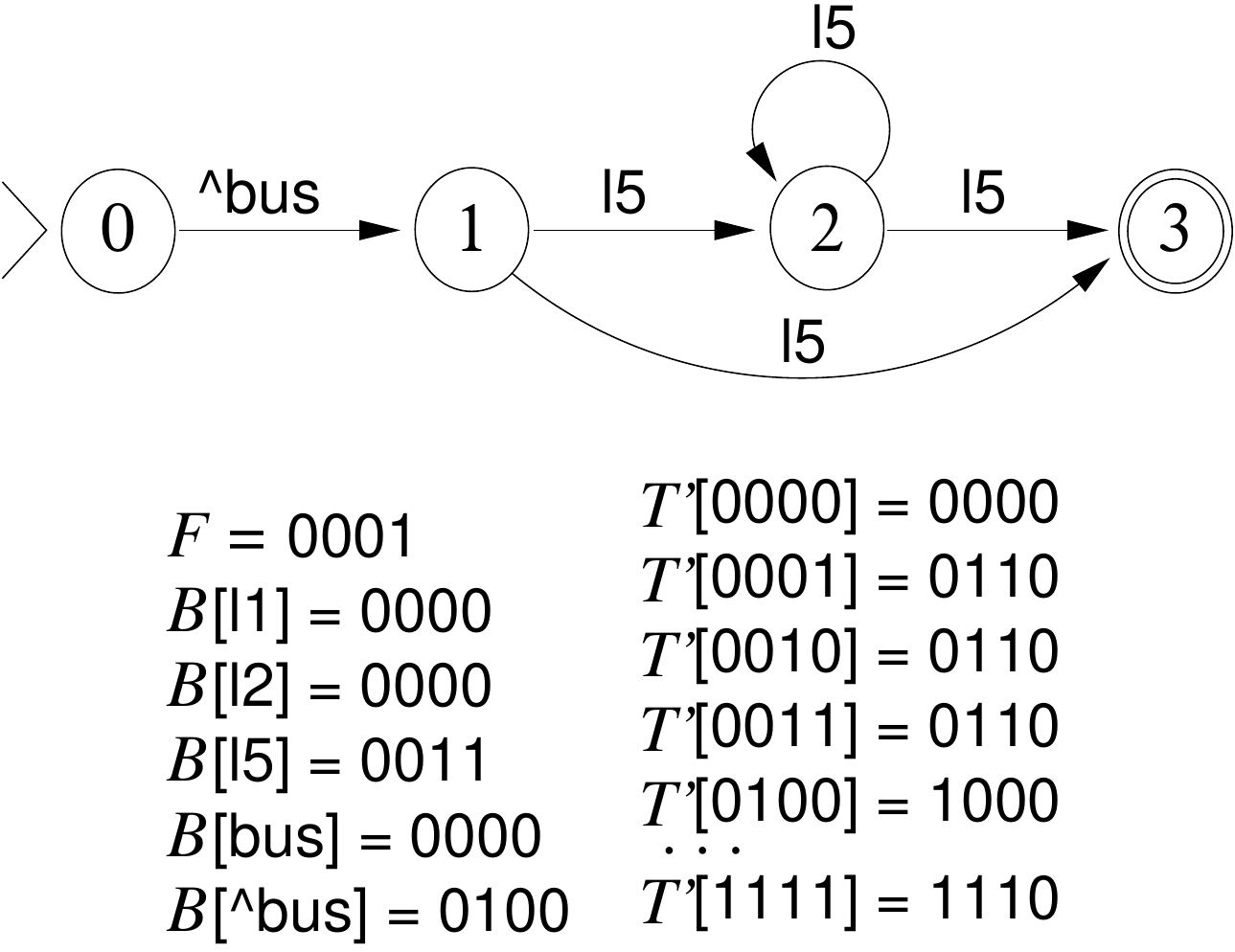}
    \caption{The Glushkov automaton for the regular expression {\rm $\textsf{\^{}bus} / \mathsf{l5}^* / \mathsf{l5}$}, its bitvector $F$ and table $B$, and the transition table $T'$ of its reversed automaton.}
    \label{fig:rpqglus}
\end{figure}

\begin{example}
Assume we are on the metro line 5 (\textsf{l5}) at station \textsf{Baq} and want to know what can we reach by following this line and then taking the bus once. The corresponding RPQ is (\textsf{Baq,l5$^+$/bus,$y$}),
and the corresponding reversed regular expression is $\inv{E} = \textsf{\^{}bus}/\mathsf{l5}^*/\mathsf{l5}$, equivalent to the example $a/b^*/b$ of Fig.~\ref{fig:glushkov}. 
 We have converted \textsf{bus} to \textsf{\^{}bus} to reverse the edge direction (we do not do the same for \textsf{l5} for simplicity, as all the metro lines bidirectional). Fig.~\ref{fig:rpqglus} shows the Glushkov automaton for this regular expression; note that
$B[\textsf{\^{}bus}]$ corresponds to $B[a]$ and $B[\textsf{l5}]$ to $B[b]$ in Fig.~\ref{fig:glushkov}, and that the alphabet of the regular expression is the set of predicates.

To solve the 2RPQ in reverse order, we start from node $5$ (\textsf{Baq}) and work backwards. We then start from $L_\mathrm{p}[C_\mathrm{o}[5]+1\dd C_\mathrm{o}[6]] = 
L_\mathrm{p}[15\dd 16]$, and report all the nodes that we can reach in reverse from there that activate the initial state of our automaton, $0$.
\end{example}

We will virtually traverse the relevant subgraph $G_E'$ of the product graph $G_E$ backwards. To simulate this process, we perform a sequence of (backward) NFA steps, traversing in reverse direction the possible paths $\rho$ that match $\inv{E}$. The traversal abandons every branch where the NFA runs out of active states. Every time it reaches the initial state we report the current node. Each NFA step starts and ends at a range of $L_\mathrm{p}$ corresponding to the current object (initially, $o$), and is simulated in three parts:
\begin{enumerate}
    \item We find all the predicates labeling edges that lead to the current object. This leads us from the interval in $L_\mathrm{p}$ (corresponding to the object) to several intervals in $L_\mathrm{s}$ (corresponding to distinct predicates for that object).
    \item We find the subjects that are sources of edges labeled with each such predicate. This leads us from each interval in $L_\mathrm{s}$ (corresponding to a predicate leading to our object) to several intervals in $L_\mathrm{o}$ (corresponding to distinct subjects).
    \item We regard each of those subjects as an object again, by mapping each resulting range in $L_\mathrm{o}$ to the corresponding range in $L_\mathrm{p}$. We only need $C_\mathrm{o}$ to do this, not $L_\mathrm{o}$.
\end{enumerate}
After steps 1 and 2, we abandon the branch if the resulting range is empty. After step 2, we perform the NFA transition and abandon the branch if $D=0$, that is, we run out of active states. We also report the subject if the initial state is active in $D$.

Note that, in step 1, we are only interested in predicates that lead to some node in $G_E'$. That is, we want predicates that lead not only to the current object, but also to active NFA states. In step 2, we are only interested in subjects that have not been visited before with the same NFA states, so as to avoid falling into loops of $G_E'$. 

In terms of the product graph, 
visiting a node $s$ of $G$ with a set $D$ of active NFA states corresponds to traversing {\em simultaneously} all the nodes of $G_E'$ that combine $s$ with an active state in $D$. Thus, bit-parallelism enables us to perform significantly less work than classical techniques that visit $G_E'$ node by node. Furthermore, we will combine Fact~\ref{fact:1} with the ability of wavelet trees to work on ranges of symbols to carry out steps 1 and 2 in a way that spends time {\em only on the resulting predicates and subjects}, thereby bounding our time complexity in terms of the subgraph of the size of $G_E'$, without spending any effort to discard edges that connect $G_E'$ with other nodes of $G_E$. We now describe each part in detail.

\subsection{Part one: Finding predicates from objects}
\label{sec:part-one}

The first part finds the distinct predicates
$p$ that lead to (i.e., precede in the $osp$ triples) the current range of objects.
We will use the wavelet tree of $L_\mathrm{p}$ to discover all the distinct 
predicates $p$ in $L_\mathrm{p}[b_o\dd e_o]$, as described at the end of Section~\ref{sec:wtrees}. From those, however, we are only interested
in those predicates $p$ that lead to a currently active NFA state. Those can be efficiently filtered thanks to Fact~\ref{fact:1}, which lets us confine the influence of $p$ to the table $B$ in the bit-parallel simulation. More precisely, by Eq.~(\ref{eq:revglus}), we are interested in the predicates $p$ such that $D ~\&~ B[p] \not= 0$. 

We will enhance the wavelet tree of $L_\mathrm{p}$ to efficiently find those predicates. We will have $B[\cdot]$ entries not only
for the predicates $p$, but also for all the other $|P|-1$ nodes in the wavelet tree of $L_\mathrm{p}$: Let $v$ be a wavelet tree node, then $B[v]$ will be the bitwise-or of
the $B[p]$ entries of all the symbols $p$ descending from $v$. 

\begin{example}
The $B[\cdot]$ entries for all the nodes of the wavelet tree of $L_\mathrm{p}$ are written as slanted bitvectors on the nodes in Fig.~\ref{fig:wt}. Those on the leaves correspond to the entries in Fig.~\ref{fig:rpqglus}, and those on internal nodes to the bitwise-or of their children.
\end{example}

This extension is easily built from the $B[p]$s in $O(m\log |P|)$ time with lazy 
initialization, by starting with all $B[v]=0$ and working upwards only from 
the nonzero entries $B[p]$, doing $B[v] \gets B[v] ~|~ B[p]$ for
every ancestor $v$ of $p$. The extra space is still $O(|P|)$, and we can
conveniently store the entries $B[v]$ in heap order, following the shape of
the (perfectly balanced) wavelet tree of $L_\mathrm{p}$.

With this extension of $B$, we proceed as follows. We start from the root $v$
of the wavelet tree of $L_\mathrm{p}$, with the range $[b\dd e] = [b_o\dd e_o]$ and
bitvector $D$. If $D ~\&~ B[v] = 0$, we stop. Otherwise, if $v$ is a leaf $p$, 
then we report the interval $L_\mathrm{s}[b\dd e]$. Otherwise, we recursively continue 
with the left and right children $v_l$ and $v_r$ of $v$, with the intervals 
$[b\dd e] = [\rank_0(W,b-1)+1\dd \rank_0(W,e)]$ for $v_l$ and
$[b\dd e] = [\rank_1(W,b-1)+1\dd \rank_1(W,e)]$ for $v_r$.

\begin{example}
To start the search from $L_\mathrm{p}[14\dd 15]$ and $D=0001$, we must first find all distinct values in the range that label transitions leading to an state active in $D$. We start from the wavelet tree root $v_{\langle 1,5\rangle}$ of Fig.~\ref{fig:wt}, with the range $L_{\langle 1,5\rangle}[14\dd 15]$. We descend to the left child, $v_{\langle 1,3\rangle}$ since $B[v_{\langle 1,3\rangle}] ~\&~ D = 0011 ~\&~ 0001 \not= 0000$ and thus there are relevant transition labels below it. When descending, we map
the range to $L_{\langle 1,3 \rangle}[9\dd 10]$ (because $\rank_0(W_{\langle 1,5\rangle},$ $14-1)+1=9$ and $\rank_0(W_{\langle 1,5\rangle},15)=10$). 
From $v_{\langle 1,3\rangle}$, we do not descend to $v_{\langle 1,2\rangle}$ since $B[v_{\langle 1,2\rangle}] ~\&~ D = 0000 ~\&~ 0001 = 0000$ and thus no relevant transition labels descend from it (though there is a $1$ in our range $L_{\langle 1,3\rangle}[9\dd 10]$ indicating an \textsf{l1} reaching \textsf{Baq}, it does not lead to active NFA states). Instead, we descend to $v_{\langle 3,3\rangle}$ because $B[v_{\langle 3,3\rangle}] ~\&~ D = 0011 ~\&~ 0001 \not= 0000$. Since it is a leaf, we have found a relevant label ($3$, i.e., \textsf{l5}) reaching our range (i.e., \textsf{Baq}). Its range is $L_{\langle 3,3\rangle}[4\dd 4]$, which added to the number of leaves in \textsf{l1} and \textsf{l2} (equivalent to $C_\mathrm{p}[3]=6$) yields the range $L_\mathrm{s}[10\dd 10]$, completing the backward search step for symbol \textsf{l5} (recall Eqs.~(\ref{eq:bwdsearch}) and (\ref{eq:bwdsearch2})).

On the other hand, we do not descend from $v_{\langle 1,5\rangle}$ to its right child, $v_{\langle 4,5 \rangle}$, because $B[v_{\langle 4,5\rangle}] ~\&~ D = 0100 ~\&~ 0001 = 0000$. Even if we did, we would obtain an empty interval in $L_{\langle 4,5\rangle}$ because there are no $4$s or $5$s in $L_{\langle 1,5\rangle}[15\dd 16]$.
\end{example}

Note that, if $D ~\&~ B[v] \not= 0$, then the same holds for at least one of the two children of $v$. As a consequence, all the wavelet tree nodes we traverse are ancestors of qualifying leaves. Since
we spend constant time on each such ancestor, we can bound the total cost of this part by charging $O(\log|P|)$ to each useful predicate $p$, for which we report the interval $L_\mathrm{s}[b_p\dd
e_p]$. We do not pay any extra cost on the useless predicates thanks to Fact~\ref{fact:1}, because we must intersect every $B[p]$ with the same set $D$ of active states.

In terms of the product graph traversal, where we are simultaneously processing all the nodes that combine $o$ with the active states in $D$; this technique allows us to obtain all the distinct edges of $G_E'$ that we can traverse from the current nodes of $G_E'$.

\subsection{Part two: Finding subjects from predicates}

The second part of the step starts at each of the ranges $L_\mathrm{s}[b_p\dd e_p]$ 
reported by the first part, and traverses the wavelet tree of $L_\mathrm{s}$ to find 
all the distinct subjects $s$ in that range, mapping them to an interval
$L_\mathrm{o}[b_s\dd e_s]$. By Fact~\ref{fact:1}, the set of active NFA states will be the same, 
$D \leftarrow T'[D ~\&~ B[p]]$ (Eq.~(\ref{eq:revglus})), for all those subjects. If $D$ contains the 
initial state, we report that subject $s$ starts a path of the 2RPQ (i.e., we report $(s,o)$ as an answer to the query).

\begin{example}
In our example, once we obtain the range $L_\mathrm{s}[10\dd 10]=4$ (\textsf{BA}) from edge label $3$ (\textsf{l5}), identifying the edge $\textsf{BA} \stackrel{\mathsf{l5}}{\longrightarrow} \textsf{Baq}$, we update $D \leftarrow T'[D ~\&~ B[3]] = T'[0001 ~\&~ 0011] = T'[0001] = 0110$, meaning we have activated states 1 and 2 in our NFA (see Fig.~\ref{fig:rpqglus}). This new state $D$ is independent of the subject we arrived at.
\end{example}

We need to prevent falling into loops, however: If we arrive at a subject 
$s$ with a subset of the NFA states we have already visited $s$ with, we must stop because we are repeating nodes
in the product graph. To
implement this filter efficiently, we will again exploit Fact~\ref{fact:1} and enhance the wavelet tree of $L_\mathrm{s}$. 

We will store for each subject $s$ a bitvector
$D[s]$ with all the active NFA states we have already reached $s$ with.
This adds $O(|V|)$ working space, but can be zeroed in constant time with lazy initialization. Thus, if we arrive at $s$ and $D ~|~ D[s] = D[s]$, the
subject $s$ can be skipped; otherwise we set $D \leftarrow D ~\& \sim\! D[s]$
and then $D[s] \leftarrow D ~|~ D[s]$, where ``$\sim$'' is the bitwise-not.
This leaves in $D$ only the NFA states that are new to $s$, and also adds to 
$D[s]$ the new active states we have arrived at $s$ with. We initially mark the states $F$ on the node $o$ where we start the search.

We use the same technique of storing $D[\cdot]$ entries at wavelet tree nodes 
$v$ of $L_\mathrm{s}$, so as to 
avoid descending by a branch if all the subjects below it have already been
visited with all the active states in $D$. For this, $D[v]$ must be the 
intersection of those $D[s]$ cells below $v$. If $D ~|~ D[v] = D[v]$, we prune the
wavelet tree traversal at node $v$, otherwise we set $D[v] \gets D ~|~ D[v]$ and continue by both left and right children. 
Just as for predicates, there is always a useful descendant leaf $s$ from the nodes $v$ we traverse, and so the total cost is $O(\log|V|)$
per useful subject arrived at.

Not visiting $s$ with a subset of the NFA states of previous visits ensures that we never work more than classical product graph traversals: every time we reprocess a node $s$ of $G$, we must be including a new NFA state in $D$ (instead, we can be faster because we handle several NFA states together, as explained). Again, we can efficiently filter the subjects with the wavelet tree thanks to Fact~\ref{fact:1}, because all the subjects $s$ are visited with the same set of states $D$.

\subsection{Part three: Mapping subjects back to objects}

In this third part, we report each
useful subject $s$ we must consider, with its corresponding state $D$. In order to proceed with the 
next step of the simulation, we
must map this range of subjects to the same range of nodes seen as objects.

This is easily done with the array $C_\mathrm{o}$, where $C_\mathrm{o}[s]$ is the number of symbols smaller than $s$ in $L_\mathrm{o}$. Thus, $L_\mathrm{p}[C_\mathrm{o}[s]+1\dd C_\mathrm{o}[s+1]]$ corresponds to the interval of $L_\mathrm{p}$ that is aligned to object $s$.

Then, as explained, we restart part one with each $s$ for which $C_\mathrm{o}[s+1]>C_\mathrm{o}[s]$,
with state $D$.

\begin{figure*}[t!]
    \centering
    \includegraphics[width=0.75\textwidth]{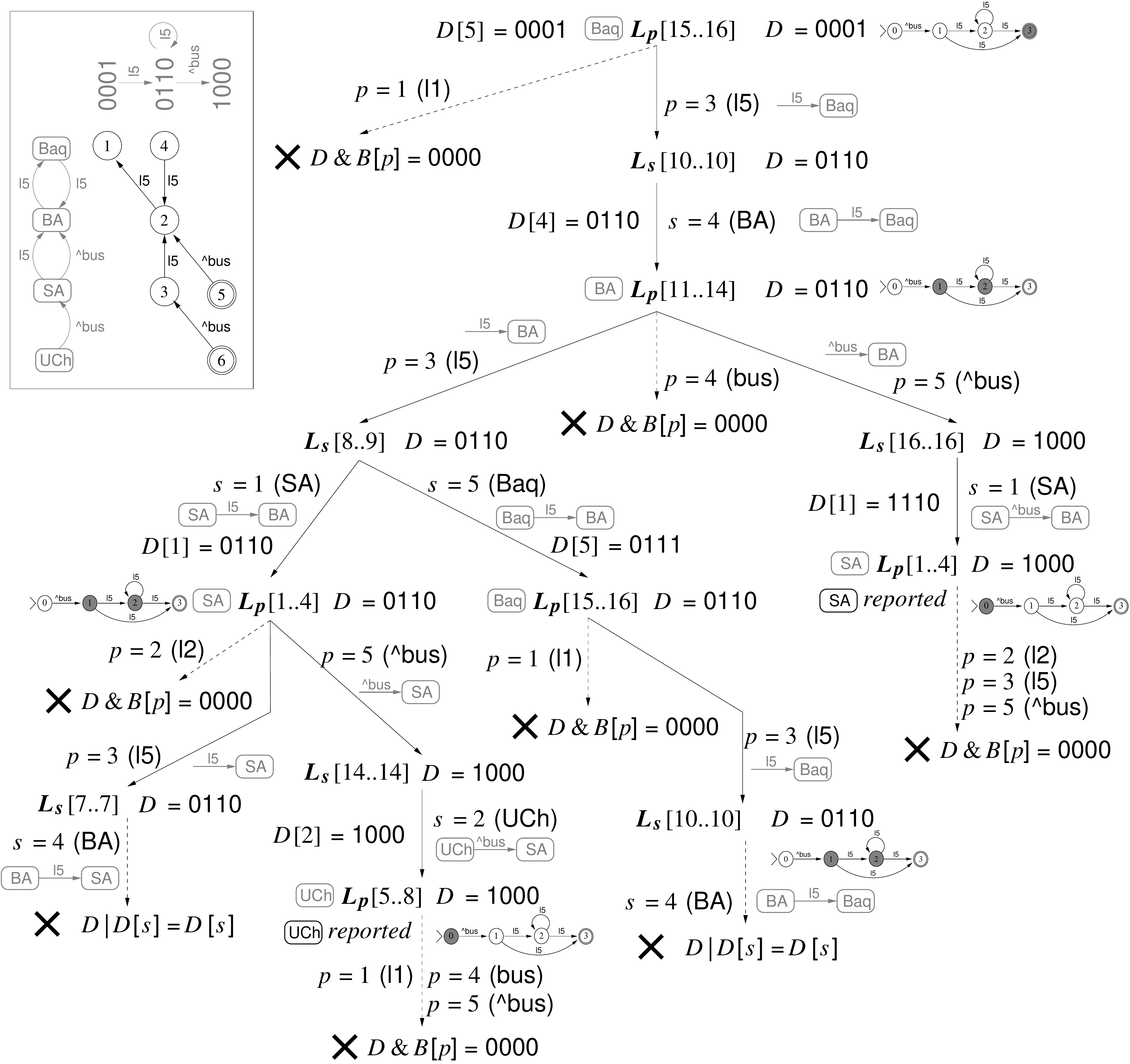}
    \caption{The whole process to match the RPQ of Fig.~\ref{fig:rpqglus} in our graph of Figs.~\ref{fig:dg2} and \ref{fig:bwt}. On the top left, the visited part of the product graph, showing in the columns the sets $D$ of NFA states we visit (and labeling the arrows for readability).}
    \label{fig:process}
\end{figure*}

\begin{example} 
Fig.~\ref{fig:process} illustrates the whole process of matching the (2)RPQ (\textsf{$y$,l5$^+$\texttt{/}\^{}bus,Baq}), with the NFA of Fig.~\ref{fig:rpqglus}, in the graph of Fig.~\ref{fig:dg2}, as represented in Fig.~\ref{fig:bwt}. We use a top-down tree to represent the branching of the process, and also show on the top left the states of the product graph $G_E'$ we traverse (backwards). We will use BFS traversal. 
The top nodes of the tree illustrate what we have already done in previous examples (edge $1 \rightarrow 2$ of $G_E'$): starting from $L_\mathrm{p}[14\dd 15]$ (\textsf{Baq}) and $D=0001$, we identified the only edge label reaching that node, \textsf{l5}, that is relevant in our NFA. Note that the label \textsf{l1} also appears in $L_\mathrm{p}[15\dd 16]$ because a transition labeled \textsf{l1} reaches \textsf{Baq}, but $D ~\&~ B[\textsf{l1}] = 0000$ because our NFA does not match it; this pruned branch is shown with a dashed arrow leading to an \textsf{X}. We have also seen that the only source of those edges labeled \textsf{l5} is $\textsf{BA}$, at $L_\mathrm{s}[10\dd 10]$, where the NFA is active at states $0110$. We do move to \textsf{BA} because $D[\textsf{BA}]=D[4]=0000$, so $D=0110$ contains some unseen NFA states at this node. We then mark $D[4]=0110$ to indicate that we have already reached $\textsf{BA}$ with those active NFA states. In part 3 of our process, we find the interval of $L_\mathrm{p}$ corresponding to $L_\mathrm{s}[10]=4$ (\textsf{BA}), $L_\mathrm{p}[C_\mathrm{o}[4]+1\dd C_\mathrm{o}[5]] = L_\mathrm{p}[11\dd 14]$, and so complete one step.

Three symbols appear on $L_\mathrm{p}[11\dd 14]$ (i.e., three edge labels reach \textsf{BA}), but only \textsf{l5} (left child) and \textsf{\^{}bus} (right child) match our NFA. By \textsf{l5} we reach $L_\mathrm{s}[8\dd 9]$ using backward search. In this interval we find two sources that, by \textsf{l5}, reach \textsf{BA}: \textsf{SA} (left child) and \textsf{Baq} (right child), both with NFA state $D=0110$ (the same as before). On the other hand, by \textsf{\^{}bus}, we reach $L_\mathrm{s}[16\dd 16]$ using backward search. There we find the only source, \textsf{SA}, that reaches \textsf{BA}, with NFA state $D=1000$. We process the three sources in BFS order, left to right:
\begin{enumerate}
\item By \textsf{l5} we reach \textsf{BA} from \textsf{SA} (leftmost tree node in this level). We accept going to \textsf{SA}, as $D[\textsf{SA}]=D[1]=0000$ and $D=0110$ has new states, so we set $D[1]=0110$. In part 3 we get the interval $L_\mathrm{p}[1\dd 4]$ for \textsf{SA}. This is edge $2 \rightarrow 3$ of $G_E'$.
\item By \textsf{l5} we reach \textsf{BA} from \textsf{Baq} (middle tree node in this level). Although we had already seen \textsf{Baq}, it was only with states $D[\textsf{Baq}]=D[5]=0001$, so the current state $D=0110$ has some unvisited NFA states; we set $D[5]=0111$ and part 3 leads us to $L_\mathrm{p}[15\dd 16]$. This is edge $2 \rightarrow 4$ of $G_E'$.
\item By \textsf{\^{}bus} we reach \textsf{BA} from \textsf{SA} as well (rightmost tree node in this level). Since $D[\textsf{SA}]=D[1]=0110$ and $D=1000$, we have new states and accept going to \textsf{SA}, setting $D[1]=1110$. The NFA state is still $1000$ (edge $2 \rightarrow 5$ of $G_E'$), which contains the initial state, so we report \textsf{SA} as a solution to our 2RPQ. We then continue from it, reaching $L_\mathrm{p}[1\dd 4]$ using part 3.
\end{enumerate}

Our BFS traversal now branches from each of those three nodes: 
\begin{enumerate}
\item From $L_\mathrm{p}[1\dd 4]$ (\textsf{SA}) with $D=0110$, we find edges labeled \textsf{l1}, \textsf{l5}, and \textsf{\^{}bus} leading to it:
\begin{enumerate}
    \item Our NFA states are not reachable by \textsf{l2} ($D ~\&~ B[\textsf{l2}]=0000$), so we abandon this branch.
    \item By \textsf{l5} we find the source \textsf{BA}, but since $D[\textsf{BA}]=D[4]=0110$ and $D=0110$, we have already visited \textsf{BA} with those active states, so we prune this branch too, avoiding a loop.
    \item By \textsf{\^{}bus} we reach $L_\mathrm{s}[14\dd 14]$ with state $D=1000$. The only source here is $L_\mathrm{s}[14]=2=$ \textsf{UCh}. Since $D[\textsf{UCh}]=D[2]=0000$, we enter this state and set $D[2]=1000$ (edge $3 \rightarrow 6$ of $G_E'$). Furthermore, since $D$ contains the initial state, we report \textsf{UCh} as the second solution to the 2RPQ.
\end{enumerate}
\item From $L_\mathrm{p}[15\dd 16]$ (\textsf{Baq}) with $D=0110$, we find edges labeled \textsf{l1} and \textsf{l5} leading to it:
\begin{enumerate}
    \item Our NFA is not interested in \textsf{l1}, so we abandon this branch.
    \item By \textsf{l5} we reach \textsf{BA} again, and again we prune the branch to avoid falling into loops, because $D[\textsf{BA}]=D[4]=0110$.
\end{enumerate}
\item From $L_\mathrm{p}[1\dd 4]$ (\textsf{SA}) and $D=1000$, which we had reported, the NFA has nowhere to go, so we reject all the possible edge labels, \textsf{l2}, \textsf{l5}, and \textsf{\^{}bus}. The same happens in the last tree level from $L_\mathrm{p}[5\dd 8]$ (\textsf{UCh}) after reporting it, so we finish.
\vspace*{-5mm}
\end{enumerate}
\end{example}

\subsection{Other kinds of RPQs}

The algorithm we have described reports all the subjects (i.e., nodes) $s \in V$
for which there is a path matching $E$ towards some object in the
range we started with. If we start with the $L_\mathrm{p}$ range for a single object
$o$ (i.e., solving $(x,E,o)$ with $x \in \V$ and $o \in V$), then the answers to the query are all the pairs $(s,o)$.
We can use this same algorithm for solving $(s,E,y)$, where $s \in V$ and $y \in \V$, by converting it into $(y,\inv{E},s)$ (as we did in our running example), so that we find all the objects $o$ reachable from a given subject $s$. 

We can also handle the case where both $s,o \in V$ are fixed, by starting from $o$ and processing $\inv{E}$, stopping as soon as we reach the state $s$, or until we run out of active states (or vice versa with $E$).

The most complex case, $(x,E,y)$ with $x,y\in\V$, has variables for both subject and object, where we must find all the pairs $(s,o)$ connected by a path matching $E$. We could handle this query by launching $|V|$ queries $(x,E,o)$, one per possible object $o$, but this would be very inefficient if many objects do not lead to answers $(s,o)$. Instead, we will use the ability of the BWT and of wavelet trees to work on ranges of symbols, not only on individual ones. Instead of starting with each specific object $o$, we will start with the full range in $L_\mathrm{p}$. Exactly the same algorithm we have described for queries $(x,E,o)$, now started with the full range, obtains all
the subjects $s$ leading to {\em some} object by $E$. Then, for every subject $s$
we arrived at, we run the RPQ $(s,E,y)$, and report $(s,o)$ for each object $o$ found in this search.
This technique guarantees that we run the queries $(s,E,y)$ only from subjects that will produce some result. Alternatively, we can first find the objects $o$ that are reachable with $E$ from some subject, and then run only the useful queries $(x,E,o)$.

The first step of this solution, which starts from the full $L_\mathrm{p}$ range, uses the power of the ring to handle a range of nodes simultaneously, effectively traversing in one step a set of nodes of $G_E'$ that relate a number of nodes of $G$ with the same set of NFA states. That is, it provides a second speedup over a classical node-wise traversal of $G_E'$. The union of the queries $(s,E,y)$ we perform amounts to a second traversal of $G_E'$.

\subsection{Time complexity}

The following theorem shows that the cost of our algorithm is essentially bounded by the size of the subgraph $G_E'$ of the product graph $G_E$ induced by the query.

\begin{theorem} \label{thm:main}
Let $G$ be a directed labeled graph over nodes $V$ and an alphabet $P$ of edge labels. Consider an RPQ $(x,E,y)$ where $x$ or $y$ are variables, $E$ has $m$ literals, and the computer word holds $O(m)$ bits. The ring representation of $G$ can return all the matching pairs $(s,o)$ for the RPQ in time $O(2^m + m\log|P|+|G_E'|\log|G|)$, where  $G_E'$ is the subgraph of the product graph $G_E$ of $G$ and Glushkov's automaton $A$ of $E$, induced by all the paths $\rho$ from any node $(s_\mu,i)$ to any node $(o_\mu,f)$, where $\mu$ is an evaluation of $(x,E,y)$, and $i$ and $f$ are initial and final states of $A$. The working space of the query is $O(m(2^m + |P| + |V| + \max_\rho))$ bits, where $\max_\rho$ is the length of the longest path $\rho$. If the computer word holds only $O(m/d)$ bits, the above times are multiplied by $d$ and $2^m$ becomes $2^{m/d}$.
\end{theorem}
\begin{proof}
The algorithm virtually visits the nodes of $G_E'$ in reverse order. Let us first consider the query $(x,E,o)$ for $o \in V$. The algorithm starts simultaneously from all the nodes $(o,f) \in G_E'$, for any final NFA state $f \in F$ (let us regard $F$ and $D$ as sets of NFA states). The algorithm preserves the invariant that, if it is at node $v \in V$ with the active NFA states $D$, then it is the first time it simulates the visit of the node $(v,d) \in G_E'$ for any NFA state $d \in D$. Each transition to new states $(v',d')$ is done in three parts. In the first part, it finds in $O(\log |P|)$ time every distinct label $p \in P$ of edges in $G_E'$ that lead to some state $(v,d)$ for $d \in D$. This cost can be charged to the edges of $G_E'$ that lead to some state $(v,d)$, because there is at least one edge per resulting label. In the second part, for each label $p$ found, it finds in $O(\log |V|)$ time every distinct node $v' \in V$ such that we reach some current node $(v,d)$ from some unvisited node $(v',d')$ in $G_E'$ via label $p$ (it obtains simultaneously the set $D'$ of all those states $d'$). We can then charge the cost to the nodes $(v',d')$ of $G_E'$. The third part takes $O(1)$ time per
node arrived at, which becomes the current node in the next iteration.

The total cost is then $O(|G_E'|\log|G|)$ for the traversal. Glushkov's construction takes $O(m^2)$ time to mark all the bits in $B[1\dd|P|]$ after the constant-time lazy initialization of $B$. The construction of the bit-parallel tables  takes time $O(2^m)$, dominated by table $T$. Computing the cells $B$ for the internal wavelet tree nodes of $L_\mathrm{p}$ adds $O(m\log|P|)$ time, again using lazy initialization, because only $O(m)$ wavelet tree leaves have nonzero cells in $B$. The lazy initialization of $D$ for the wavelet tree nodes of $L_\mathrm{s}$ adds $O(1)$ time. 

The working space is $O(m2^m)$ bits for the bit-parallel simulation of the NFA, $O(m(|P|+|V|))$ bits for the tables $B$/$D$ on the wavelet tree nodes of $L_\mathrm{p}$/$L_\mathrm{s}$, $O(|P|+|V|)$ bits for the compact structures for the lazy initialization of the tables $B/D$ \cite[App.~C]{Nav14b}, and $O(m\max_\rho)$ bits for the recursive path traversals carrying the active states $D$.

If the computer word cannot hold $O(m)$ bits, we split the words into $O(d)$ words of $O(m/d)$ bits each and operate on each word sequentially, thereby multiplying times by $O(d)$.

All the other types of queries $(x,E,y)$ are reduced to the query $(x,E,o)$ we have considered: in case $(s,E,y)$, where $s \in V$ and $y \in \V$, we just reverse the query; in case $(x,E,y)$ where $x,y \in \V$, we perform many queries $(s,E,y)$, which subsume the cost of the initial query that finds the relevant nodes $s$ from any $o$. Here $G_E'$ is the union of the graphs $G_E'$ for every $s$. The only case where we may visit more nodes than those in $G_E'$ is the query $(s,E,o)$ with both $s,e \in V$, because we visit nodes that may not lead to $s$. This is why this case is left out of the theorem.
\end{proof}

The theorem does not capture the fact that we are able to process several NFA states $d$ simultaneously when traversing the nodes $(v,d)$ of $G_E'$, thanks to the bit-parallel simulation of the NFA. 

\begin{example}
Fig.~\ref{fig:product2} shows the product graph $G_E$ of our running example, highlighting in bold the nodes and edges of $G_E'$. The dashed edges also belong to $G_E'$ but we avoid them to prevent loops. The shaded nodes correspond to the reported results. Theorem~\ref{thm:main} proves that we spend, at worst, logarithmic time per node and edge of $G_E'$. Comparing the figure with the top-left part of Fig.~\ref{fig:process}, however, one can see that our simulation processes the nodes of the second and third column of $G_E$ simultaneously (in column $0110$ of Fig.~\ref{fig:process}). We maintained in Fig.~\ref{fig:product2} the numbering of the nodes of Fig.~\ref{fig:process}, which helps see the nodes we visit simultaneously.
\end{example}

\begin{figure}[t!]
    \centering
    \includegraphics[width=0.22\textwidth]{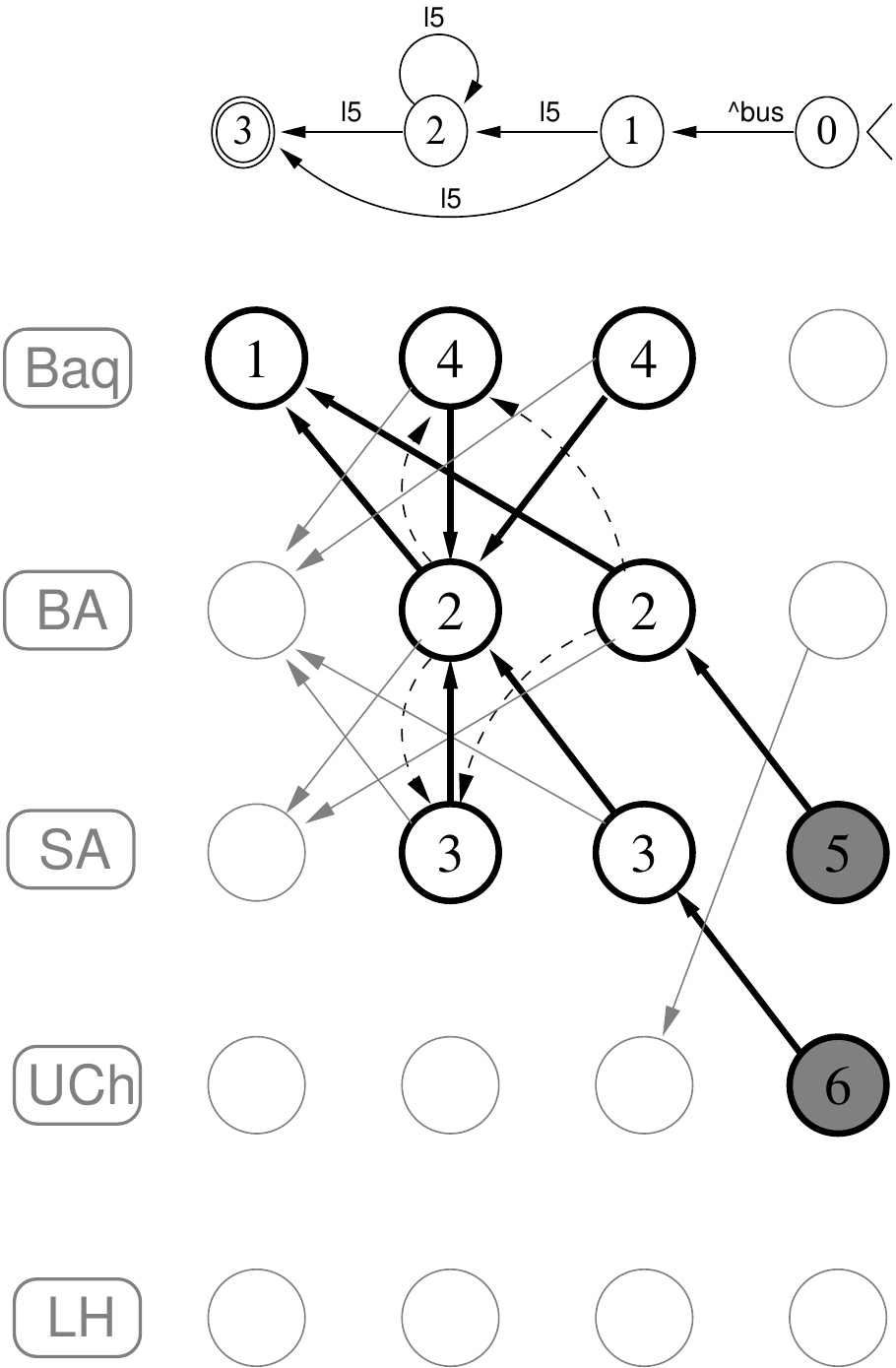}
    \caption{The product graph $G_E$ of the graph of Fig.~\ref{fig:dg2} and the automaton of Fig.~\ref{fig:rpqglus}, highlighting in black the subgraph $G_E'$ that is traversed (backwards) in Fig.~\ref{fig:process}.}
    \label{fig:product2}
\end{figure}

\section{Implementation and Experiments}

We implemented our scheme in C++11 using the succinct data structures library (\textsc{sdsl}, \url{https://github.com/simongog/sdsl-lite}). 
We ran our experiments on an Intel(R) Xeon(R) CPU E5-2630 at 2.30GHz, with 6 cores, 15 MB of cache, and 96 GB of RAM. Our code was compiled using \texttt{g++} with flags \texttt{-std=c++11}, \texttt{-O3}, and \texttt{-msse4.2}. 

\paragraph{Benchmark.}
We test our approach on a real-world benchmark: a Wikidata graph~\cite{VrandecicK14} of $n = 958{,}844{,}164$ edges and $|V| = 348{,}945{,}080$ nodes. This graph has
$|S| = 106{,}736{,}662$ subjects, $|P| = 5{,}419$ predicates, and $|O| = 295{,}611{,}216$ objects. This dataset occupies 10.7 GB in plain form (i.e., with 32-bit integers for each triple component, or 12 bytes per tuple) and 7.9 GB in packed form (i.e., using $\lceil \log |S|\rceil+\lceil\log |P|\rceil+\lceil\log|O|\rceil$ bits, or 8.625 bytes, per triple). 

We compare with the following graph database systems, in terms of the space used for indexing and the time for evaluating RPQs:\footnote{While ArangoDB, Neo4j, OrientDB and TigerGraph support various path features -- such as bounded traversals, shortest paths, Kleene star, paths with per-edge filters, etc. -- to the best of our knowledge they do not support RPQs declaratively.}
\begin{description}
\item[Jena:] A reference implementation of the SPARQL standard. 
\item[Virtuoso:] A widely used graph database hosting the public DBpedia endpoint, among others~\cite{Virtuoso}. 
\item[Blazegraph:] The graph database system~\cite{Blazegraph} hosting the official Wikidata Query Service~\cite{MalyshevKGGB18}. 
\end{description}

Systems are then configured per vendor recommendations, as in previous work~\cite{AHNRRS21}. Jena, Virtuoso and Blazegraph all implement RPQs per the semantics of property paths in SPARQL~1.1, whereby fixed-length paths (without \texttt{*} or \texttt{+}) are translated into SPARQL graph patterns without RPQs and evaluated under bag semantics. All systems apply set semantics for arbitrary-length paths, per the SPARQL standard. Jena and Blazegraph implement a navigational BFS-style function called ALP (Arbitrary Length Paths) defined by the SPARQL standard~\cite{sparql11}, while Virtuoso uses a transitive closure operator implemented over its relational database engine.

In order to test on challenging, real-world RPQs, we extracted all non-trivial RPQs (i.e., not a simple label) from the code-\texttt{500} (timeout) sections of all seven intervals of the the Wikidata Query Logs~\cite{MalyshevKGGB18}. After filtering RPQs mentioning constants not used in the dataset, normalizing variable names, and removing duplicates, this process yielded 1{,}952 unique queries. All queries are run with a timeout of 60 seconds under set semantics (using \texttt{DISTINCT} in the case of SPARQL) with a limit of 1 million results for comparability (Virtuoso has a hard-coded limit of $2^{20}$ results). 
We classify the RPQs of our log into patterns by mapping nodes to constant/variable types and erasing their predicates (keeping only RPQ operators); for example, ($x$, $p_1/p_2^*$, $y$) has the pattern $\mathtt{c~/*~v}$, $\mathtt{v~/*~c}$, or $\mathtt{v~/*~v}$, depending on whether $x$ and $y$ are constant ($\mathtt{c}$) or variable ($\mathtt{v}$). Table~\ref{tab:rpq-patterns} shows the 20 most popular RPQ patterns in our log.  

\begin{table}
\centering
\caption{The 20 most popular RPQ patterns in our query log.}
\label{tab:rpq-patterns}

\begin{tabular}{rr}
\toprule
\multicolumn{1}{l}{1st--7th} & \multicolumn{1}{c}{\#}\\
\midrule
\texttt{v /* c} & 537\\
\texttt{v * c} & 433\\
\texttt{v + c} & 109\\
\texttt{c * v} & 99\\
\texttt{c /* v} & 95\\
\texttt{v / c} & 54\\
\texttt{v */* c} & 44\\
\bottomrule
\end{tabular}
\hfill
\begin{tabular}{rr}
\toprule
\multicolumn{1}{l}{8th--14th} & \multicolumn{1}{c}{\#}\\
\midrule
 \texttt{v / v} & 41\\ 
\texttt{v |* c} & 36\\ 
\texttt{v | v} & 31\\
\texttt{v */*/*/*/* c} & 28\\
\texttt{v \^{} v} & 26\\
\texttt{v /* v} & 25\\
\texttt{v * v} & 25\\
\bottomrule
\end{tabular}
\hfill
\begin{tabular}{rr}
\toprule
\multicolumn{1}{l}{15th--20th} & \multicolumn{1}{c}{\#}\\
\midrule
\texttt{v /? c} & 22\\
\texttt{v + v} & 17\\
\texttt{v /+ c} & 12\\
 \texttt{v || v} & 10 \\ 
\texttt{v | c} & 10\\ 
\texttt{v /\^{} v} & 7\\
\ \\
\bottomrule
\end{tabular}

\no{
\begin{tabular}{rr}
\toprule
Path pattern (1--10) & \multicolumn{1}{c}{Count}\\
\midrule
\texttt{v /* c} & 537\\
\texttt{v * c} & 433\\
\texttt{v + c} & 109\\
\texttt{c * v} & 99\\
\texttt{c /* v} & 95\\
\texttt{v / c} & 54\\
\texttt{v */* c} & 44\\
 \texttt{v / v} & 41\\ 
\texttt{v |* c} & 36\\ 
\texttt{v | v} & 31\\
\bottomrule
\end{tabular}
\hfill
\begin{tabular}{rr}
\toprule
Path pattern (11-20) & \multicolumn{1}{c}{Count} \\
\midrule
\texttt{v */*/*/*/* c} & 28\\
\texttt{v \^{} v} & 26\\
\texttt{v /* v} & 25\\
\texttt{v * v} & 25\\
\texttt{v /? c} & 22\\
\texttt{v + v} & 17\\
\texttt{v /+ c} & 12\\
 \texttt{v || v} & 10 \\ 
\texttt{v | c} & 10\\ 
\texttt{v /\^{} v} & 7\\
\bottomrule
\end{tabular}
}
\end{table}

\paragraph{Index construction}
We work with a dictionary-encoded version of the graph as described in Section~\ref{sec:our-approach}, where we complete the graph by adding the reversed edges with inverse labels: If an edge is labeled with predicate $p$, its reverse edge has predicate $\inv{p}=p + |P|$. This doubles the number of edges and the number of predicates.
To construct our index, we build arrays $L_{\mathrm{s}}$ and $L_{\mathrm{p}}$ (and the corresponding $C_{\mathrm{p}}$ and $C_{\mathrm{o}}$) using a suffix array \cite{AHNRRS21}. We represent $L_{\mathrm{s}}$ and $L_{\mathrm{p}}$ using wavelet matrices \cite{CNO15}, a particular implementation of wavelet trees to handle big alphabets efficiently. We use plain bitvectors to implement the wavelet-matrix nodes. Array $C_{\mathrm{o}}$ is represented using a plain bitvector, whereas $C_{\mathrm{p}}$ is represented as a simple array. 
Our index is constructed in 2.3 hours, using 64.75 GB of RAM.

\paragraph{Implementing queries}
We use our generic query algorithm of Section~\ref{sec:our-approach}, but handle
the query patterns $\texttt{v\^{}v}$, $\texttt{v/\^{}v}$, $\texttt{v|v}$, $\texttt{v||v}$, and $\texttt{v/v}$ more efficiently using just backward search and the extended functionality of wavelet trees:
For a variable-to-variable query $(x, p, y)$ (analogously, $(x, \inv{p}, y)$), we start by extracting all subjects $s$ from $L_{\mathrm{s}}[C_{\mathrm{p}}[p]{..}C_{\mathrm{p}}[p+1]-1]$, using the wavelet tree. Then, for each $s$ in that range, we start at range
$[C_{\mathrm{o}}[s]{..}C_{\mathrm{o}}[s]-1]$ in $L_{\mathrm{p}}$ and carry out a backward search step using $\inv{p}$. This yields the range of $L_{\mathrm{s}}$ containing all values $o$ such that $(s, p, o)$ is a graph edge, so we report $(s, o)$.
Query $(x, p_1|p_2, y)$ (similarly, $(x, p_2| p_3| p_4, y)$) is decomposed into queries $(x, p_1, y)$ and $(x, p_2, y)$, which are computed as explained before. To detect duplicate pairs $(s,o)$, we use a hash table (\texttt{std::unordered\_set} in C++).
For query $(x, p_1/p_2, y)$ (similarly, $(x, p_1/\inv{p_2}, y)$) we first find all nodes $z$ that are the target of an edge labeled $p_1$, and the origin of an edge labeled $p_2$. This is done by intersecting the ranges $L_{\mathrm{s}}[C_{\mathrm{p}}[\inv{p_1}]{..}C_{\mathrm{p}}[\inv{p_1}+1]-1]$ and
$L_{\mathrm{s}}[C_{\mathrm{p}}[p_2]{..}C_{\mathrm{p}}[p_2+1]-1]$, using the wavelet tree capabilities \cite{GNP12}. Then, for every such $z$ in the intersection, we carry out a backward search for $p_1 z$, to find all nodes $s$ such that $(s, p_1, z)$ is a graph edge. Similarly, we do a backward search for $\inv{p_2} z$, to find all nodes $o$ such that $(z, p_2, o)$ is a graph edge. Then, for every such $s$ and $o$ we report $(s, o)$, again avoiding duplicates.
Finally, for queries $(x, p_1/(p_2)^*, y)$ we start the search always from $p_1$. In general, this filters candidates more efficiently For all the remaining queries $(x, E, y)$, we choose to start from the end whose predicate has the smallest cardinality.

We implement array $B$  (used to filter on $L_{\mathrm{p}}$ in Section \ref{sec:part-one}) with an array of integers, initially zeroed. We do lazy initialization by setting the values of the different 
predicates of the query and their wavelet matrix ancestors, and zeroing them again after running the query.
Array $D$, on the other hand, is implemented using the a compact lazy initialization structure \cite[App.~C]{Nav14b}, which uses $O(|V|)$ extra bits on top of $D$. We use 16-bit cells for $D$, as queries in our log have fewer than 16 predicates (with some few exceptions that use operator \texttt{|}, which are handled differently as explained).

 

\begin{table}
\centering
\caption{Index space (in bytes per edge) and some statistics on the query times (in seconds). Row ``Timeouts'' counts the queries taking over 60 seconds.}
\label{tab:statistics}
%
%
\begin{tabular}{l@{\hskip 5mm}rrrr}
\toprule
  & Ring &  Jena & Virtuoso & Blazegraph \\
\midrule
Space & 16.41 & 95.83 & 60.07 & 90.79\\ 
\midrule
Average  & 3.73 & 9.93 & 9.15 & 6.23 \\
Median & 0.15 & 0.42 & 0.32 & 0.15 \\
Timeouts & 43 & 273 & 154 & 135 \\
\midrule
Average $\mathtt{c}$-to-$\mathtt{v}$ & 0.99 & 4.99 & 4.59 & 4.37\\
Median $\mathtt{c}$-to-$\mathtt{v}$ & 0.07 & 0.22 & 0.15 & 0.14 \\
\midrule
Average $\mathtt{v}$-to-$\mathtt{v}$ & 18.97 & 37.46 & 34.54 & 16.63\\
Median $\mathtt{v}$-to-$\mathtt{v}$ & 8.48 & 60.00 & 28.62 & 0.21\\
\bottomrule
\end{tabular}
\end{table}

\paragraph{Space and query time}

Table \ref{tab:statistics} shows the space usage of the systems we tested, as well as statistics about the whole query process.
The ring is the smallest index, using 16.41 bytes per triple. This is about twice the space of the compact representation of the data, consistent with the fact that we duplicate all the edges. Array $D$, needed at query time, uses 3.09 additional bytes per triple, whereas $B$ uses $9.04\times10^{-5}$ bytes per triple. The total working space usage at query time is 19.50 bytes per triple, $1/3$--$1/5$ of the space used by the other indexes (not considering their extra working space).

Regarding query time, the ring is the fastest on average, being 1.67 times faster than Blazegraph, the next best performer. The ring is also the index with fewest timeouts. 
On the queries where some node is a constant (``$\mathtt{c}$-to-$\mathtt{v}$'' in the table, 84.7\% of the log), the ring is 4.41 times faster than Blazegraph on average, with only 1 timeout (no system was able to complete that query). The other systems have over 50 timeouts on these queries.
For the queries where both nodes are variables (``$\mathtt{v}$-to-$\mathtt{v}$'', 15.3\% of the log), the ring is the second best, only 14\% slower than Blazegraph on average. 

Fig.~\ref{fig:query-time-1} shows the distribution of query times for the different patterns.
\begin{figure}[pt!]
\centering
    \includegraphics[width=\columnwidth]{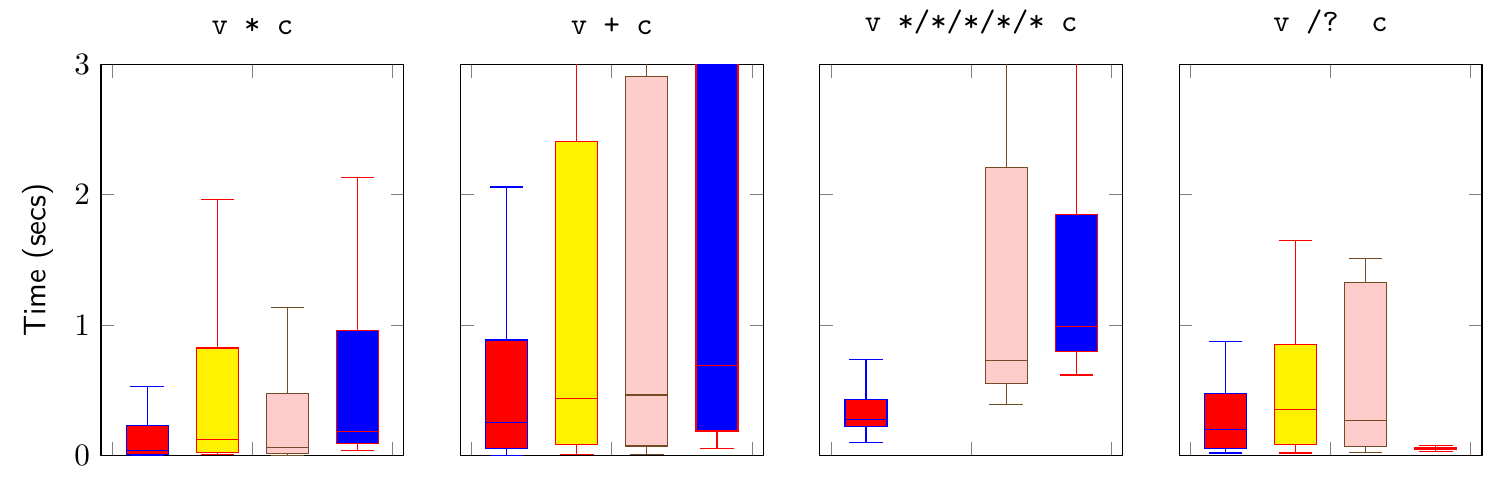}
    \includegraphics[width=\columnwidth]{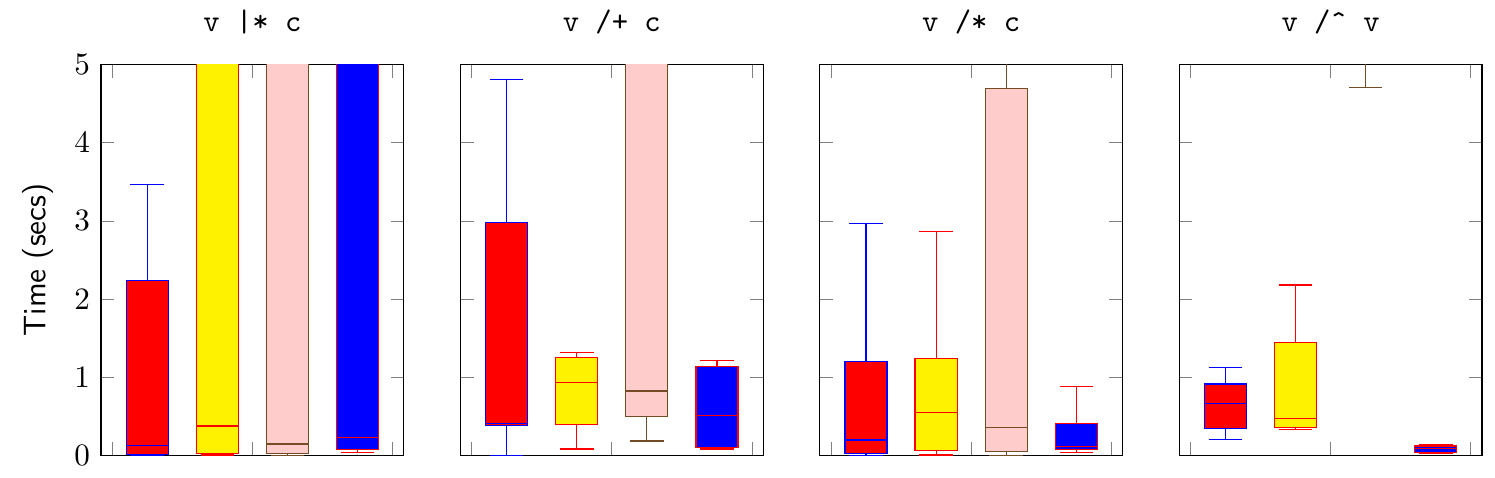}
    \includegraphics[width=\columnwidth]{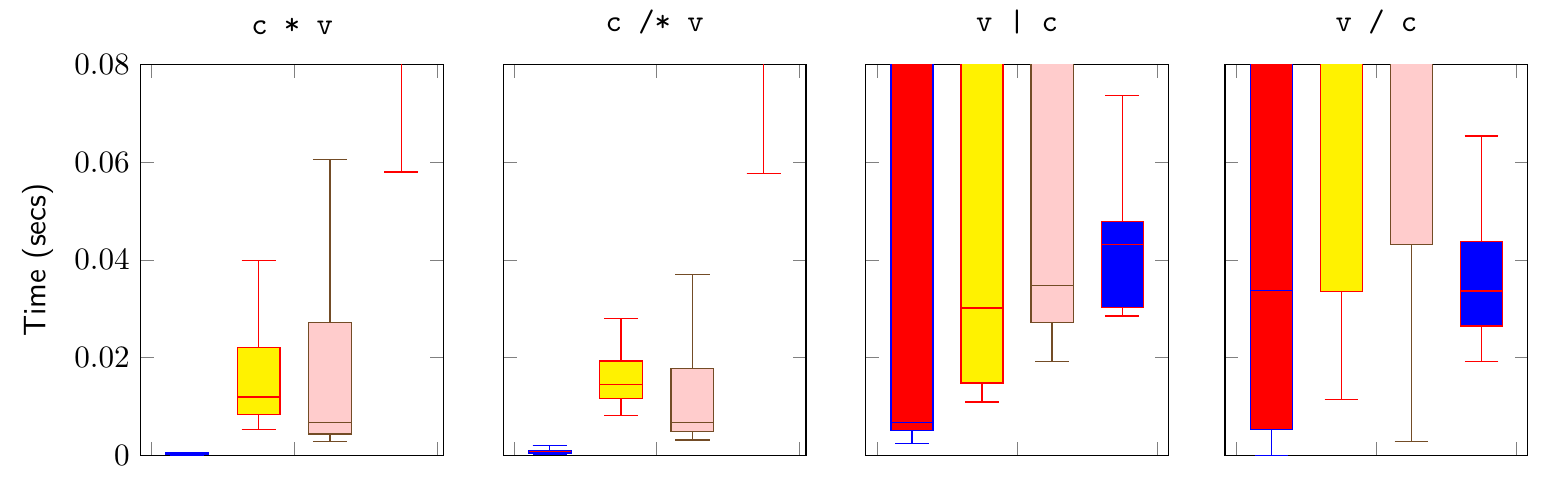}
    \includegraphics[width=\columnwidth]{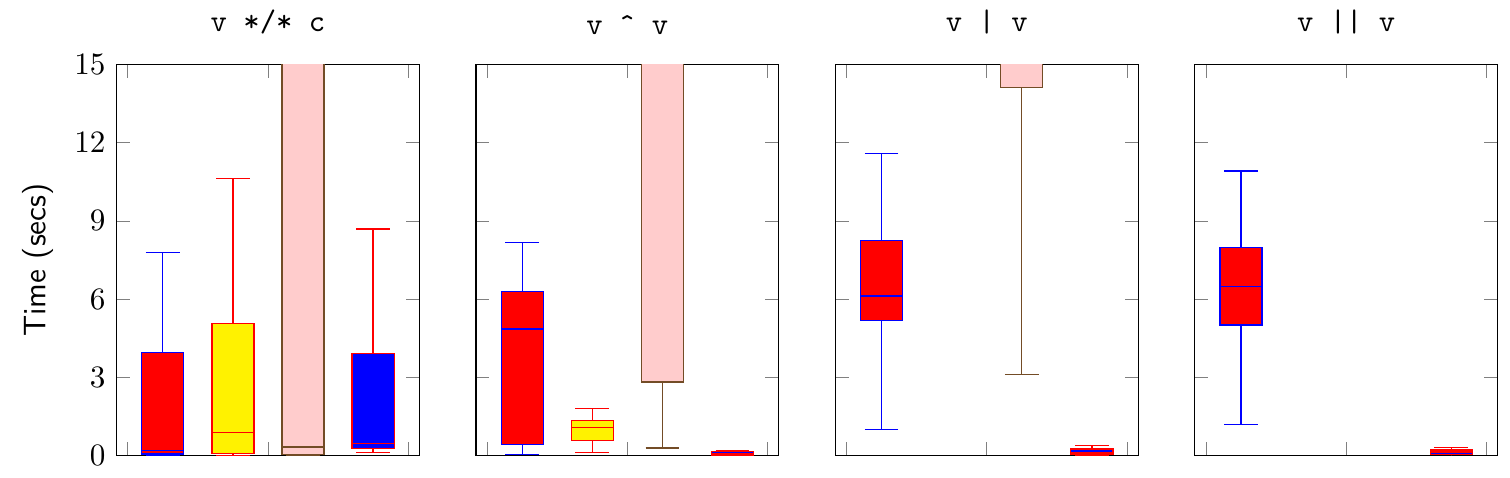}
    \includegraphics[width=\columnwidth]{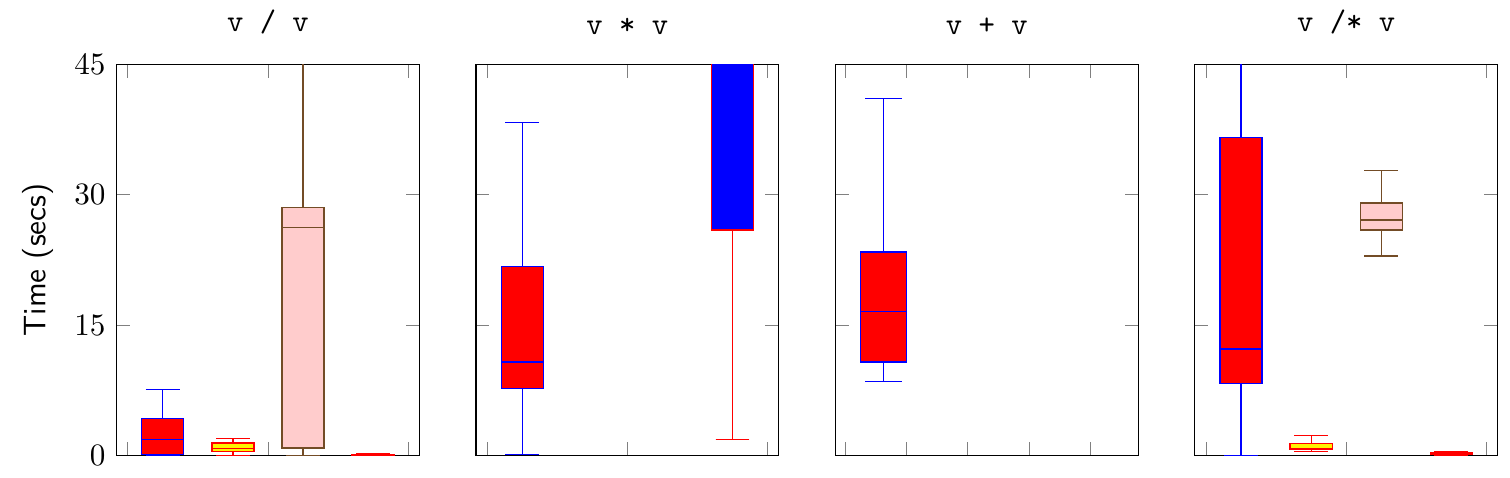}

    \vspace{0.3cm}
    \begin{tikzpicture}
       \begin{customlegend}[legend columns=4,legend style={draw=none,column sep=1ex, font=\footnotesize, cells={anchor=west}, at={(0.5,-0.15)},
      anchor=north, /tikz/every even column/.append style={column sep=0.2cm}},legend entries={\small\textsf{Ring},\small\textsf{Jena}, \small\textsf{Virtuoso}, 
       \small\textsf{Blazegraph}}]
          \addlegendimage{blue,fill=red, ybar stacked}
          \addlegendimage{red,fill=yellow,ybar stacked}
          \addlegendimage{brown,fill=red!20,ybar stacked}
          \addlegendimage{red,fill=blue,ybar stacked}
       \end{customlegend}
    \end{tikzpicture}
    \caption{Boxplots for the distribution of query times.} 
    \label{fig:query-time-1}
\end{figure}
Our approach has the best performance in 9 out of 20 patterns, which correspond to 
45.39\% of the query log. 
Each of these 9 patterns have at least one $*$ or $+$. There are only other 3 such patterns, \texttt{v/+c}, \texttt{v/*c}, and \texttt{v/*v}, on which our ring is outperformed by another system. The remaining 8 patterns on which we do not compete are paths of length 1, length 2, and ``or'' paths. Such paths can be solved as join queries, with more efficient algorithms. 

Note that our system works on the integer-encoded triples, whereas others work on the original string values. As shown in previous work \cite{AHNRRS21}, we can encode the strings of this benchmark within just 3 additional bytes per triple and incurring in around 3 extra milliseconds per query, in order to decode the answers.

\section{Conclusions}

We have shown how the ring \cite{AHNRRS21}, a compact representation of labeled graphs developed to support worst-case optimal graph joins, can be enhanced by combining in a unique way the capabilities of (1) the wavelet trees, to process ranges of graph nodes or labels, and (2) the bit-parallel simulation of Glushkov automata, to handle various NFA states simultaneously, in order to solve regular path queries (RPQs) on the graph. We prove that the cost of the resulting algorithm is proportional to the subgraph of the product graph induced by the query, but our technique is even faster because it is able to process groups of nodes and labels simultaneously. As a result, our index uses 3--5 times less space than the alternatives, while matching or exceeding their performance (on average, our index is the fastest, outperforming the next best by a factor of 1.67).

We have not yet explored strategies for partitioning the NFA at edges with labels that appear infrequently in the graph and then joining the results, as do several techniques described in Section~\ref{sec:related}. Our techniques do permit running the NFA forwards or backwards from those labels, so this could be explored in future. Furthermore, the wavelet tree offers powerful operations that provide on-the-fly selectivity statistics, which can be used for even more sophisticated query planning. For example, by roughly doubling the space, we can compute in logarithmic time the amount of distinct predicates labeling edges towards a given range of objects, or distinct subjects that are sources of a given range of predicates \cite{GKNP13}.

%

Our technique is particularly well-suited to integrate RPQs in SPARQL multijoin queries solved with Leapfrog Triejoin, reusing the same ring data structure \cite{AHNRRS21}. In this case, in addition to the triples of the basic graph patterns, there will be triples of the
form $(x,E,y)$, where $E$ is a regular expression. By treating $E$ as any other relation, the Leapfrog algorithm will 
choose to first instantiate $x$ (resp., $y$), and thus will ask for the smallest $x \ge x_0$ (resp., $y \ge y_0$) that has a solution for some $y$ (resp., $x$). Later, it will instantiate $y$ (resp., $x$) and will ask for the smallest $y \ge y_0$ (resp, $x \ge x_0)$ that has a solution for a concrete value of $x$ (resp., $y$). The capability of wavelet trees to work on ranges of symbols allows us to find those smallest $x \ge x_0$ or $y \ge y_0$ values efficiently, for example by successive binary partitioning of the range of candidates. 

Other requirements 
are also efficiently met with our data structures. For example, we can easily enforce visiting specific nodes within the regular expression, or that those nodes have specific attribute values, by marking the noncomplying nodes as already visited with the NFA states that enforce those conditions, so they will be avoided in our traversal. The bit-parallel Glushkov simulation also efficiently handles classes of symbols labeling the NFA edges (like $(\gf{l1}|\gf{l2}|\gf{l5})$, or negated labels), without building unnecessarily large NFAs; this could be used to support negated property sets defined in SPARQL property paths, or even inference over RDF graphs (e.g., handling virtual disjunctions of inferred properties). 


\vspace*{-2mm}
\begin{acks}
This work was supported by ANID – Millennium Science Initiative Program – Code ICN17\_002.
\end{acks}

\bibliographystyle{ACM-Reference-Format}
\bibliography{paper}


\begin{thebibliography}{56}


\ifx \showCODEN    \undefined \def \showCODEN     #1{\unskip}     \fi
\ifx \showDOI      \undefined \def \showDOI       #1{#1}\fi
\ifx \showISBNx    \undefined \def \showISBNx     #1{\unskip}     \fi
\ifx \showISBNxiii \undefined \def \showISBNxiii  #1{\unskip}     \fi
\ifx \showISSN     \undefined \def \showISSN      #1{\unskip}     \fi
\ifx \showLCCN     \undefined \def \showLCCN      #1{\unskip}     \fi
\ifx \shownote     \undefined \def \shownote      #1{#1}          \fi
\ifx \showarticletitle \undefined \def \showarticletitle #1{#1}   \fi
\ifx \showURL      \undefined \def \showURL       {\relax}        \fi
\providecommand\bibfield[2]{#2}
\providecommand\bibinfo[2]{#2}
\providecommand\natexlab[1]{#1}
\providecommand\showeprint[2][]{arXiv:#2}

\bibitem[\protect\citeauthoryear{Abul{-}Basher}{Abul{-}Basher}{2017}]%
        {Abul-Basher17}
\bibfield{author}{\bibinfo{person}{Zahid Abul{-}Basher}.}
  \bibinfo{year}{2017}\natexlab{}.
\newblock \showarticletitle{{Multiple-Query Optimization of Regular Path
  Queries}}. In \bibinfo{booktitle}{\emph{International Conference on Data
  Engineering (ICDE)}}. \bibinfo{publisher}{{IEEE} Computer Society},
  \bibinfo{pages}{1426--1430}.
\newblock


\bibitem[\protect\citeauthoryear{Angles, Arenas, Barcel{\'{o}}, Boncz,
  Fletcher, Guti{\'{e}}rrez, Lindaaker, Paradies, Plantikow, Sequeda, van Rest,
  and Voigt}{Angles et~al\mbox{.}}{2018}]%
        {AnglesABBFGLPPS18}
\bibfield{author}{\bibinfo{person}{Renzo Angles}, \bibinfo{person}{Marcelo
  Arenas}, \bibinfo{person}{Pablo Barcel{\'{o}}}, \bibinfo{person}{Peter~A.
  Boncz}, \bibinfo{person}{George H.~L. Fletcher}, \bibinfo{person}{Claudio
  Guti{\'{e}}rrez}, \bibinfo{person}{Tobias Lindaaker}, \bibinfo{person}{Marcus
  Paradies}, \bibinfo{person}{Stefan Plantikow}, \bibinfo{person}{Juan~F.
  Sequeda}, \bibinfo{person}{Oskar van Rest}, {and} \bibinfo{person}{Hannes
  Voigt}.} \bibinfo{year}{2018}\natexlab{}.
\newblock \showarticletitle{{G-CORE:} {A} Core for Future Graph Query
  Languages}. In \bibinfo{booktitle}{\emph{SIGMOD International Conference on
  Management of Data}}. \bibinfo{publisher}{{ACM}},
  \bibinfo{pages}{1421--1432}.
\newblock
\urldef\tempurl%
\url{https://doi.org/10.1145/3183713.3190654}
\showDOI{\tempurl}


\bibitem[\protect\citeauthoryear{Angles, Arenas, Barcel{\'{o}}, Hogan, Reutter,
  and Vrgoc}{Angles et~al\mbox{.}}{2017}]%
        {AnglesABHRV17}
\bibfield{author}{\bibinfo{person}{Renzo Angles}, \bibinfo{person}{Marcelo
  Arenas}, \bibinfo{person}{Pablo Barcel{\'{o}}}, \bibinfo{person}{Aidan
  Hogan}, \bibinfo{person}{Juan~L. Reutter}, {and} \bibinfo{person}{Domagoj
  Vrgoc}.} \bibinfo{year}{2017}\natexlab{}.
\newblock \showarticletitle{{Foundations of Modern Query Languages for Graph
  Databases}}.
\newblock \bibinfo{journal}{\emph{{ACM} Comput. Surv.}} \bibinfo{volume}{50},
  \bibinfo{number}{5} (\bibinfo{year}{2017}), \bibinfo{pages}{68:1--68:40}.
\newblock
\urldef\tempurl%
\url{https://doi.org/10.1145/3104031}
\showDOI{\tempurl}


\bibitem[\protect\citeauthoryear{Arroyuelo, Hogan, Navarro, Reutter,
  Rojas-Ledesma, and Soto}{Arroyuelo et~al\mbox{.}}{2021}]%
        {AHNRRS21}
\bibfield{author}{\bibinfo{person}{Diego Arroyuelo}, \bibinfo{person}{Aidan
  Hogan}, \bibinfo{person}{Gonzalo Navarro}, \bibinfo{person}{Juan Reutter},
  \bibinfo{person}{Javiel Rojas-Ledesma}, {and} \bibinfo{person}{Adri{\'a}n
  Soto}.} \bibinfo{year}{2021}\natexlab{}.
\newblock \showarticletitle{Worst-Case Optimal Graph Joins in Almost No Space}.
  In \bibinfo{booktitle}{\emph{ACM International Conference on Management of
  Data (SIGMOD)}}. \bibinfo{pages}{102--114}.
\newblock


\bibitem[\protect\citeauthoryear{Baier, Daroch, Reutter, and Vrgoc}{Baier
  et~al\mbox{.}}{2017}]%
        {BaierDRV17}
\bibfield{author}{\bibinfo{person}{Jorge~A. Baier}, \bibinfo{person}{Dietrich
  Daroch}, \bibinfo{person}{Juan~L. Reutter}, {and} \bibinfo{person}{Domagoj
  Vrgoc}.} \bibinfo{year}{2017}\natexlab{}.
\newblock \showarticletitle{{Evaluating Navigational RDF Queries over the
  Web}}. In \bibinfo{booktitle}{\emph{{ACM} Conference on Hypertext and Social
  Media (HT)}}. \bibinfo{publisher}{{ACM}}, \bibinfo{pages}{165--174}.
\newblock


\bibitem[\protect\citeauthoryear{Berry and Sethi}{Berry and Sethi}{1986}]%
        {BS86}
\bibfield{author}{\bibinfo{person}{Gerard Berry} {and} \bibinfo{person}{Ravi
  Sethi}.} \bibinfo{year}{1986}\natexlab{}.
\newblock \showarticletitle{From regular expression to deterministic automata}.
\newblock \bibinfo{journal}{\emph{Theoretical Computer Science}}
  \bibinfo{volume}{48}, \bibinfo{number}{1} (\bibinfo{year}{1986}),
  \bibinfo{pages}{117--126}.
\newblock


\bibitem[\protect\citeauthoryear{Bonifati, Martens, and Timm}{Bonifati
  et~al\mbox{.}}{2019}]%
        {BonifatiMT19}
\bibfield{author}{\bibinfo{person}{Angela Bonifati}, \bibinfo{person}{Wim
  Martens}, {and} \bibinfo{person}{Thomas Timm}.}
  \bibinfo{year}{2019}\natexlab{}.
\newblock \showarticletitle{{Navigating the Maze of Wikidata Query Logs}}. In
  \bibinfo{booktitle}{\emph{The World Wide Web Conference (WWW)}}.
  \bibinfo{publisher}{{ACM}}, \bibinfo{pages}{127--138}.
\newblock


\bibitem[\protect\citeauthoryear{Br{\"u}ggemann-Klein}{Br{\"u}ggemann-Klein}{1993}]%
        {BK93}
\bibfield{author}{\bibinfo{person}{Anne Br{\"u}ggemann-Klein}.}
  \bibinfo{year}{1993}\natexlab{}.
\newblock \showarticletitle{Regular expressions into finite automata}.
\newblock \bibinfo{journal}{\emph{Theoretical Computer Science}}
  \bibinfo{volume}{120}, \bibinfo{number}{2} (\bibinfo{year}{1993}),
  \bibinfo{pages}{197--213}.
\newblock


\bibitem[\protect\citeauthoryear{Burrows and Wheeler}{Burrows and
  Wheeler}{1994}]%
        {BW94}
\bibfield{author}{\bibinfo{person}{Michael Burrows} {and}
  \bibinfo{person}{David Wheeler}.} \bibinfo{year}{1994}\natexlab{}.
\newblock \bibinfo{booktitle}{\emph{A block sorting lossless data compression
  algorithm}}.
\newblock \bibinfo{type}{{T}echnical {R}eport} 124.
  \bibinfo{institution}{Digital Equipment Corporation}.
\newblock


\bibitem[\protect\citeauthoryear{Clark}{Clark}{1996}]%
        {Cla96}
\bibfield{author}{\bibinfo{person}{David~R. Clark}.}
  \bibinfo{year}{1996}\natexlab{}.
\newblock \emph{\bibinfo{title}{Compact {PAT} Trees}}.
\newblock \bibinfo{thesistype}{Ph.D. Dissertation}. \bibinfo{school}{University
  of Waterloo, Canada}.
\newblock


\bibitem[\protect\citeauthoryear{Claude, Navarro, and Ord{\'o}{\~n}ez}{Claude
  et~al\mbox{.}}{2015}]%
        {CNO15}
\bibfield{author}{\bibinfo{person}{Francisco Claude}, \bibinfo{person}{Gonzalo
  Navarro}, {and} \bibinfo{person}{Alberto Ord{\'o}{\~n}ez}.}
  \bibinfo{year}{2015}\natexlab{}.
\newblock \showarticletitle{The wavelet matrix: An efficient wavelet tree for
  large alphabets}.
\newblock \bibinfo{journal}{\emph{Information Systems}}  \bibinfo{volume}{47}
  (\bibinfo{year}{2015}), \bibinfo{pages}{15--32}.
\newblock


\bibitem[\protect\citeauthoryear{Colazzo, Mecca, Nol{\'{e}}, and
  Sartiani}{Colazzo et~al\mbox{.}}{2018}]%
        {ColazzoMNS18}
\bibfield{author}{\bibinfo{person}{Dario Colazzo}, \bibinfo{person}{Vincenzo
  Mecca}, \bibinfo{person}{Maurizio Nol{\'{e}}}, {and} \bibinfo{person}{Carlo
  Sartiani}.} \bibinfo{year}{2018}\natexlab{}.
\newblock \showarticletitle{{PathGraph: querying and exploring big data
  graphs}}. In \bibinfo{booktitle}{\emph{International Conference on Scientific
  and Statistical Database Management (SSDBM)}}. \bibinfo{publisher}{{ACM}},
  \bibinfo{pages}{29:1--29:4}.
\newblock


\bibitem[\protect\citeauthoryear{Cruz, Mendelzon, and Wood}{Cruz
  et~al\mbox{.}}{1987}]%
        {CruzMW87}
\bibfield{author}{\bibinfo{person}{Isabel~F. Cruz}, \bibinfo{person}{Alberto~O.
  Mendelzon}, {and} \bibinfo{person}{Peter~T. Wood}.}
  \bibinfo{year}{1987}\natexlab{}.
\newblock \showarticletitle{{A Graphical Query Language Supporting Recursion}}.
  In \bibinfo{booktitle}{\emph{SIGMOD International Conference on Management of
  Data}}. \bibinfo{publisher}{{ACM} Press}, \bibinfo{pages}{323--330}.
\newblock


\bibitem[\protect\citeauthoryear{Dey, Cuevas{-}Vicentt{\'{\i}}n, K{\"{o}}hler,
  Gribkoff, Wang, and Lud{\"{a}}scher}{Dey et~al\mbox{.}}{2013}]%
        {DeyC0GWL13}
\bibfield{author}{\bibinfo{person}{Saumen~C. Dey},
  \bibinfo{person}{V{\'{\i}}ctor Cuevas{-}Vicentt{\'{\i}}n},
  \bibinfo{person}{Sven K{\"{o}}hler}, \bibinfo{person}{Eric Gribkoff},
  \bibinfo{person}{Michael Wang}, {and} \bibinfo{person}{Bertram
  Lud{\"{a}}scher}.} \bibinfo{year}{2013}\natexlab{}.
\newblock \showarticletitle{On implementing provenance-aware regular path
  queries with relational query engines}. In \bibinfo{booktitle}{\emph{Joint
  2013 {EDBT/ICDT} Conferences}}. \bibinfo{publisher}{{ACM}},
  \bibinfo{pages}{214--223}.
\newblock


\bibitem[\protect\citeauthoryear{Erling and Mikhailov}{Erling and
  Mikhailov}{2009}]%
        {Virtuoso}
\bibfield{author}{\bibinfo{person}{Orri Erling} {and} \bibinfo{person}{Ivan
  Mikhailov}.} \bibinfo{year}{2009}\natexlab{}.
\newblock \showarticletitle{{RDF support in the Virtuoso DBMS}}.
\newblock In \bibinfo{booktitle}{\emph{Networked Knowledge -- Networked
  Media}}. \bibinfo{publisher}{Springer}, \bibinfo{pages}{7--24}.
\newblock


\bibitem[\protect\citeauthoryear{Ferragina and Manzini}{Ferragina and
  Manzini}{2005}]%
        {FM05}
\bibfield{author}{\bibinfo{person}{Paolo Ferragina} {and}
  \bibinfo{person}{Giovanni Manzini}.} \bibinfo{year}{2005}\natexlab{}.
\newblock \showarticletitle{Indexing compressed texts}.
\newblock \bibinfo{journal}{\emph{{Journal of the ACM}}} \bibinfo{volume}{52},
  \bibinfo{number}{4} (\bibinfo{year}{2005}), \bibinfo{pages}{552--581}.
\newblock


\bibitem[\protect\citeauthoryear{Fionda, Pirr{\`{o}}, and Consens}{Fionda
  et~al\mbox{.}}{2019}]%
        {FiondaPC19}
\bibfield{author}{\bibinfo{person}{Valeria Fionda}, \bibinfo{person}{Giuseppe
  Pirr{\`{o}}}, {and} \bibinfo{person}{Mariano~P. Consens}.}
  \bibinfo{year}{2019}\natexlab{}.
\newblock \showarticletitle{{Querying knowledge graphs with extended property
  paths}}.
\newblock \bibinfo{journal}{\emph{Semantic Web}} \bibinfo{volume}{10},
  \bibinfo{number}{6} (\bibinfo{year}{2019}), \bibinfo{pages}{1127--1168}.
\newblock


\bibitem[\protect\citeauthoryear{Fletcher, Peters, and Poulovassilis}{Fletcher
  et~al\mbox{.}}{2016}]%
        {FletcherPP16}
\bibfield{author}{\bibinfo{person}{George H.~L. Fletcher},
  \bibinfo{person}{Jeroen Peters}, {and} \bibinfo{person}{Alexandra
  Poulovassilis}.} \bibinfo{year}{2016}\natexlab{}.
\newblock \showarticletitle{{Efficient regular path query evaluation using path
  indexes}}. In \bibinfo{booktitle}{\emph{International Conference on Extending
  Database Technology (EDBT)}}. \bibinfo{publisher}{OpenProceedings.org},
  \bibinfo{pages}{636--639}.
\newblock


\bibitem[\protect\citeauthoryear{Francis, Green, Guagliardo, Libkin, Lindaaker,
  Marsault, Plantikow, Rydberg, Selmer, and Taylor}{Francis
  et~al\mbox{.}}{2018}]%
        {FrancisGGLLMPRS18}
\bibfield{author}{\bibinfo{person}{Nadime Francis}, \bibinfo{person}{Alastair
  Green}, \bibinfo{person}{Paolo Guagliardo}, \bibinfo{person}{Leonid Libkin},
  \bibinfo{person}{Tobias Lindaaker}, \bibinfo{person}{Victor Marsault},
  \bibinfo{person}{Stefan Plantikow}, \bibinfo{person}{Mats Rydberg},
  \bibinfo{person}{Petra Selmer}, {and} \bibinfo{person}{Andr{\'{e}}s Taylor}.}
  \bibinfo{year}{2018}\natexlab{}.
\newblock \showarticletitle{{Cypher: An Evolving Query Language for Property
  Graphs}}. In \bibinfo{booktitle}{\emph{SIGMOD International Conference on
  Management of Data}}. \bibinfo{publisher}{{ACM}},
  \bibinfo{pages}{1433--1445}.
\newblock


\bibitem[\protect\citeauthoryear{Gagie, K{\"a}rkk{\"a}inen, Navarro, and
  Puglisi}{Gagie et~al\mbox{.}}{2013}]%
        {GKNP13}
\bibfield{author}{\bibinfo{person}{Travis Gagie}, \bibinfo{person}{Juha
  K{\"a}rkk{\"a}inen}, \bibinfo{person}{Gonzalo Navarro}, {and}
  \bibinfo{person}{Simon~J. Puglisi}.} \bibinfo{year}{2013}\natexlab{}.
\newblock \showarticletitle{Colored range queries and document retrieval}.
\newblock \bibinfo{journal}{\emph{Theoretical Computer Science}}
  \bibinfo{volume}{483} (\bibinfo{year}{2013}), \bibinfo{pages}{36--50}.
\newblock


\bibitem[\protect\citeauthoryear{Gagie, Navarro, and Puglisi}{Gagie
  et~al\mbox{.}}{2012}]%
        {GNP12}
\bibfield{author}{\bibinfo{person}{Travis Gagie}, \bibinfo{person}{Gonzalo
  Navarro}, {and} \bibinfo{person}{Simon~J. Puglisi}.}
  \bibinfo{year}{2012}\natexlab{}.
\newblock \showarticletitle{New Algorithms on Wavelet Trees and Applications to
  Information Retrieval}.
\newblock \bibinfo{journal}{\emph{Theoretical Computer Science}}
  \bibinfo{volume}{426-427} (\bibinfo{year}{2012}), \bibinfo{pages}{25--41}.
\newblock


\bibitem[\protect\citeauthoryear{Glushkov}{Glushkov}{1961}]%
        {Glu61}
\bibfield{author}{\bibinfo{person}{V.-M. Glushkov}.}
  \bibinfo{year}{1961}\natexlab{}.
\newblock \showarticletitle{The abstract theory of automata}.
\newblock \bibinfo{journal}{\emph{Russian Mathematical Surveys}}
  \bibinfo{volume}{16} (\bibinfo{year}{1961}), \bibinfo{pages}{1--53}.
\newblock


\bibitem[\protect\citeauthoryear{Grossi, Gupta, and Vitter}{Grossi
  et~al\mbox{.}}{2003}]%
        {GGV03}
\bibfield{author}{\bibinfo{person}{Roberto Grossi}, \bibinfo{person}{Ankur
  Gupta}, {and} \bibinfo{person}{Jeff~S. Vitter}.}
  \bibinfo{year}{2003}\natexlab{}.
\newblock \showarticletitle{High-order entropy-compressed text indexes}. In
  \bibinfo{booktitle}{\emph{Proc. 14th Annual ACM-SIAM Symposium on Discrete
  Algorithms (SODA)}}. \bibinfo{pages}{841--850}.
\newblock


\bibitem[\protect\citeauthoryear{Gubichev, Bedathur, and Seufert}{Gubichev
  et~al\mbox{.}}{2013}]%
        {GubichevBS13}
\bibfield{author}{\bibinfo{person}{Andrey Gubichev},
  \bibinfo{person}{Srikanta~J. Bedathur}, {and} \bibinfo{person}{Stephan
  Seufert}.} \bibinfo{year}{2013}\natexlab{}.
\newblock \showarticletitle{{Sparqling kleene: fast property paths in RDF-3X}}.
  In \bibinfo{booktitle}{\emph{International Workshop on Graph Data Management
  Experiences and Systems (GRADES)}}. \bibinfo{publisher}{{CWI/ACM}},
  \bibinfo{pages}{14}.
\newblock


\bibitem[\protect\citeauthoryear{Gubichev and Neumann}{Gubichev and
  Neumann}{2011}]%
        {GubichevN11}
\bibfield{author}{\bibinfo{person}{Andrey Gubichev} {and}
  \bibinfo{person}{Thomas Neumann}.} \bibinfo{year}{2011}\natexlab{}.
\newblock \showarticletitle{{Path Query Processing on Very Large RDF Graphs}}.
  In \bibinfo{booktitle}{\emph{International Workshop on the Web and Databases
  (WebDB)}}.
\newblock


\bibitem[\protect\citeauthoryear{Guo, Gao, and Zou}{Guo et~al\mbox{.}}{2021}]%
        {GuoGZ21}
\bibfield{author}{\bibinfo{person}{Xintong Guo}, \bibinfo{person}{Hong Gao},
  {and} \bibinfo{person}{Zhaonian Zou}.} \bibinfo{year}{2021}\natexlab{}.
\newblock \showarticletitle{{Distributed processing of regular path queries in
  RDF graphs}}.
\newblock \bibinfo{journal}{\emph{Knowl. Inf. Syst.}} \bibinfo{volume}{63},
  \bibinfo{number}{4} (\bibinfo{year}{2021}), \bibinfo{pages}{993--1027}.
\newblock


\bibitem[\protect\citeauthoryear{Harris, Seaborne, and Prud'hommeaux}{Harris
  et~al\mbox{.}}{2013}]%
        {sparql11}
\bibfield{author}{\bibinfo{person}{Steve Harris}, \bibinfo{person}{Andy
  Seaborne}, {and} \bibinfo{person}{Eric Prud'hommeaux}.}
  \bibinfo{year}{2013}\natexlab{}.
\newblock \bibinfo{title}{{SPARQL 1.1 Query Language}}.
\newblock \bibinfo{howpublished}{{W3C Recommendation}}.
\newblock
\newblock
\shownote{\url{http://www.w3.org/TR/sparql11-query/}.}


\bibitem[\protect\citeauthoryear{Hartig and Pirr{\`{o}}}{Hartig and
  Pirr{\`{o}}}{2017}]%
        {HartigP17}
\bibfield{author}{\bibinfo{person}{Olaf Hartig} {and} \bibinfo{person}{Giuseppe
  Pirr{\`{o}}}.} \bibinfo{year}{2017}\natexlab{}.
\newblock \showarticletitle{{SPARQL with property paths on the Web}}.
\newblock \bibinfo{journal}{\emph{Semantic Web}} \bibinfo{volume}{8},
  \bibinfo{number}{6} (\bibinfo{year}{2017}), \bibinfo{pages}{773--795}.
\newblock
\urldef\tempurl%
\url{https://doi.org/10.3233/SW-160237}
\showDOI{\tempurl}


\bibitem[\protect\citeauthoryear{Jachiet, Genev{\`{e}}s, Gesbert, and
  Laya{\"{\i}}da}{Jachiet et~al\mbox{.}}{2020}]%
        {JachietGGL20}
\bibfield{author}{\bibinfo{person}{Louis Jachiet}, \bibinfo{person}{Pierre
  Genev{\`{e}}s}, \bibinfo{person}{Nils Gesbert}, {and} \bibinfo{person}{Nabil
  Laya{\"{\i}}da}.} \bibinfo{year}{2020}\natexlab{}.
\newblock \showarticletitle{{On the Optimization of Recursive Relational
  Queries: Application to Graph Queries}}. In \bibinfo{booktitle}{\emph{SIGMOD
  International Conference on Management of Data (SIGMOD)}}.
  \bibinfo{publisher}{{ACM}}, \bibinfo{pages}{681--697}.
\newblock


\bibitem[\protect\citeauthoryear{Koschmieder and Leser}{Koschmieder and
  Leser}{2012}]%
        {KoschmiederL12}
\bibfield{author}{\bibinfo{person}{Andr{\'{e}} Koschmieder} {and}
  \bibinfo{person}{Ulf Leser}.} \bibinfo{year}{2012}\natexlab{}.
\newblock \showarticletitle{{Regular Path Queries on Large Graphs}}. In
  \bibinfo{booktitle}{\emph{International Conference on Scientific and
  Statistical Database Management (SSDBM)}} \emph{(\bibinfo{series}{LNCS})},
  Vol.~\bibinfo{volume}{7338}. \bibinfo{publisher}{Springer},
  \bibinfo{pages}{177--194}.
\newblock


\bibitem[\protect\citeauthoryear{Kostylev, Reutter, Romero, and Vrgoc}{Kostylev
  et~al\mbox{.}}{2015}]%
        {KostylevR0V15}
\bibfield{author}{\bibinfo{person}{Egor~V. Kostylev}, \bibinfo{person}{Juan~L.
  Reutter}, \bibinfo{person}{Miguel Romero}, {and} \bibinfo{person}{Domagoj
  Vrgoc}.} \bibinfo{year}{2015}\natexlab{}.
\newblock \showarticletitle{{SPARQL with Property Paths}}. In
  \bibinfo{booktitle}{\emph{International Semantic Web Conference (ISWC)}}
  \emph{(\bibinfo{series}{LNCS})}, Vol.~\bibinfo{volume}{9366}.
  \bibinfo{publisher}{Springer}, \bibinfo{pages}{3--18}.
\newblock


\bibitem[\protect\citeauthoryear{Kuijpers, Fletcher, Lindaaker, and
  Yakovets}{Kuijpers et~al\mbox{.}}{2021}]%
        {KuijpersFLY21}
\bibfield{author}{\bibinfo{person}{Jochem Kuijpers}, \bibinfo{person}{George
  Fletcher}, \bibinfo{person}{Tobias Lindaaker}, {and} \bibinfo{person}{Nikolay
  Yakovets}.} \bibinfo{year}{2021}\natexlab{}.
\newblock \showarticletitle{{Path Indexing in the Cypher Query Pipeline}}. In
  \bibinfo{booktitle}{\emph{International Conference on Extending Database
  Technology (EDBT)}}. \bibinfo{publisher}{OpenProceedings.org},
  \bibinfo{pages}{582--587}.
\newblock


\bibitem[\protect\citeauthoryear{Liu, Wang, Liu, Li, and Wang}{Liu
  et~al\mbox{.}}{2021}]%
        {LiuWLLW21}
\bibfield{author}{\bibinfo{person}{Baozhu Liu}, \bibinfo{person}{Xin Wang},
  \bibinfo{person}{Pengkai Liu}, \bibinfo{person}{Sizhuo Li}, {and}
  \bibinfo{person}{Xiaofei Wang}.} \bibinfo{year}{2021}\natexlab{}.
\newblock \showarticletitle{{PAIRPQ: An Efficient Path Index for Regular Path
  Queries on Knowledge Graphs}}. In \bibinfo{booktitle}{\emph{International
  Joint Conference on Web and Big Data (APWeb-WAIM)}}
  \emph{(\bibinfo{series}{LNCS})}, Vol.~\bibinfo{volume}{12859}.
  \bibinfo{publisher}{Springer}, \bibinfo{pages}{106--120}.
\newblock


\bibitem[\protect\citeauthoryear{Malyshev, Kr{\"{o}}tzsch, Gonz{\'{a}}lez,
  Gonsior, and Bielefeldt}{Malyshev et~al\mbox{.}}{2018}]%
        {MalyshevKGGB18}
\bibfield{author}{\bibinfo{person}{Stanislav Malyshev}, \bibinfo{person}{Markus
  Kr{\"{o}}tzsch}, \bibinfo{person}{Larry Gonz{\'{a}}lez},
  \bibinfo{person}{Julius Gonsior}, {and} \bibinfo{person}{Adrian Bielefeldt}.}
  \bibinfo{year}{2018}\natexlab{}.
\newblock \showarticletitle{{Getting the Most Out of Wikidata: Semantic
  Technology Usage in Wikipedia's Knowledge Graph}}. In
  \bibinfo{booktitle}{\emph{International Semantic Web Conference (ISWC)}}.
  \bibinfo{pages}{376--394}.
\newblock


\bibitem[\protect\citeauthoryear{Mehmood, Saleem, Sahay, Ngomo, and
  d'Aquin}{Mehmood et~al\mbox{.}}{2019}]%
        {MehmoodSSNd19}
\bibfield{author}{\bibinfo{person}{Qaiser Mehmood}, \bibinfo{person}{Muhammad
  Saleem}, \bibinfo{person}{Ratnesh Sahay},
  \bibinfo{person}{Axel{-}Cyrille~Ngonga Ngomo}, {and} \bibinfo{person}{Mathieu
  d'Aquin}.} \bibinfo{year}{2019}\natexlab{}.
\newblock \showarticletitle{{QPPDs: Querying Property Paths Over Distributed
  RDF Datasets}}.
\newblock \bibinfo{journal}{\emph{{IEEE} Access}}  \bibinfo{volume}{7}
  (\bibinfo{year}{2019}), \bibinfo{pages}{101031--101045}.
\newblock


\bibitem[\protect\citeauthoryear{Mendelzon and Wood}{Mendelzon and
  Wood}{1995}]%
        {MendelzonW95}
\bibfield{author}{\bibinfo{person}{Alberto~O. Mendelzon} {and}
  \bibinfo{person}{Peter~T. Wood}.} \bibinfo{year}{1995}\natexlab{}.
\newblock \showarticletitle{Finding Regular Simple Paths in Graph Databases}.
\newblock \bibinfo{journal}{\emph{SIAM J. Comput.}} \bibinfo{volume}{24},
  \bibinfo{number}{6} (\bibinfo{year}{1995}), \bibinfo{pages}{1235--1258}.
\newblock


\bibitem[\protect\citeauthoryear{Miao, Stefanescu, and Thomo}{Miao
  et~al\mbox{.}}{2007}]%
        {MiaoST07}
\bibfield{author}{\bibinfo{person}{Zhuo Miao}, \bibinfo{person}{Dan~C.
  Stefanescu}, {and} \bibinfo{person}{Alex Thomo}.}
  \bibinfo{year}{2007}\natexlab{}.
\newblock \showarticletitle{{Grid-Aware Evaluation of Regular Path Queries on
  Spatial Networks}}. In \bibinfo{booktitle}{\emph{International Conference on
  Advanced Information Networking and Applications (AINA)}}.
  \bibinfo{publisher}{{IEEE} Computer Society}, \bibinfo{pages}{158--165}.
\newblock


\bibitem[\protect\citeauthoryear{Miura, Amagasa, and Kitagawa}{Miura
  et~al\mbox{.}}{2019}]%
        {MiuraAK19}
\bibfield{author}{\bibinfo{person}{Kento Miura}, \bibinfo{person}{Toshiyuki
  Amagasa}, {and} \bibinfo{person}{Hiroyuki Kitagawa}.}
  \bibinfo{year}{2019}\natexlab{}.
\newblock \showarticletitle{{Accelerating Regular Path Queries using FPGA}}. In
  \bibinfo{booktitle}{\emph{International Workshop on Accelerating Analytics
  and Data Management Systems Using Modern Processor and Storage Architectures
  (ADMS@VLDB)}}, \bibfield{editor}{\bibinfo{person}{Rajesh Bordawekar} {and}
  \bibinfo{person}{Tirthankar Lahiri}} (Eds.). \bibinfo{pages}{47--54}.
\newblock


\bibitem[\protect\citeauthoryear{Munro}{Munro}{1996}]%
        {Mun96}
\bibfield{author}{\bibinfo{person}{J.~Ian Munro}.}
  \bibinfo{year}{1996}\natexlab{}.
\newblock \showarticletitle{Tables}. In \bibinfo{booktitle}{\emph{Proc. 16th
  Conference on Foundations of Software Technology and Theoretical Computer
  Science (FSTTCS)}}. \bibinfo{pages}{37--42}.
\newblock


\bibitem[\protect\citeauthoryear{Navarro}{Navarro}{2014a}]%
        {Nav14b}
\bibfield{author}{\bibinfo{person}{Gonzalo Navarro}.}
  \bibinfo{year}{2014}\natexlab{a}.
\newblock \showarticletitle{Spaces, Trees and Colors: The Algorithmic Landscape
  of Document Retrieval on Sequences}.
\newblock \bibinfo{journal}{\emph{Comput. Surveys}} \bibinfo{volume}{46},
  \bibinfo{number}{4} (\bibinfo{year}{2014}), \bibinfo{pages}{article 52}.
\newblock


\bibitem[\protect\citeauthoryear{Navarro}{Navarro}{2014b}]%
        {Nav14}
\bibfield{author}{\bibinfo{person}{Gonzalo Navarro}.}
  \bibinfo{year}{2014}\natexlab{b}.
\newblock \showarticletitle{Wavelet Trees for All}.
\newblock \bibinfo{journal}{\emph{Journal of Discrete Algorithms}}
  \bibinfo{volume}{25} (\bibinfo{year}{2014}), \bibinfo{pages}{2--20}.
\newblock


\bibitem[\protect\citeauthoryear{Navarro and Raffinot}{Navarro and
  Raffinot}{2005}]%
        {NR04}
\bibfield{author}{\bibinfo{person}{Gonzalo Navarro} {and}
  \bibinfo{person}{Mathieu Raffinot}.} \bibinfo{year}{2005}\natexlab{}.
\newblock \showarticletitle{New Techniques for Regular Expression Searching}.
\newblock \bibinfo{journal}{\emph{Algorithmica}} \bibinfo{volume}{41},
  \bibinfo{number}{2} (\bibinfo{year}{2005}), \bibinfo{pages}{89--116}.
\newblock


\bibitem[\protect\citeauthoryear{Nguyen and Kim}{Nguyen and Kim}{2017}]%
        {NguyenK17}
\bibfield{author}{\bibinfo{person}{Van{-}Quyet Nguyen} {and}
  \bibinfo{person}{Kyungbaek Kim}.} \bibinfo{year}{2017}\natexlab{}.
\newblock \showarticletitle{{Efficient Regular Path Query Evaluation by
  Splitting with Unit-Subquery Cost Matrix}}.
\newblock \bibinfo{journal}{\emph{{IEICE} Trans. Inf. Syst.}}
  \bibinfo{volume}{100-D}, \bibinfo{number}{10} (\bibinfo{year}{2017}),
  \bibinfo{pages}{2648--2652}.
\newblock


\bibitem[\protect\citeauthoryear{Nol{\'{e}} and Sartiani}{Nol{\'{e}} and
  Sartiani}{2016}]%
        {NoleS16}
\bibfield{author}{\bibinfo{person}{Maurizio Nol{\'{e}}} {and}
  \bibinfo{person}{Carlo Sartiani}.} \bibinfo{year}{2016}\natexlab{}.
\newblock \showarticletitle{{Regular Path Queries on Massive Graphs}}. In
  \bibinfo{booktitle}{\emph{SIGMOD International Conference on Scientific and
  Statistical Database Management (SSDBM)}}. \bibinfo{publisher}{{ACM}},
  \bibinfo{pages}{13:1--13:12}.
\newblock


\bibitem[\protect\citeauthoryear{Pacaci, Bonifati, and {\"{O}}zsu}{Pacaci
  et~al\mbox{.}}{2020}]%
        {PacaciBO20}
\bibfield{author}{\bibinfo{person}{Anil Pacaci}, \bibinfo{person}{Angela
  Bonifati}, {and} \bibinfo{person}{M.~Tamer {\"{O}}zsu}.}
  \bibinfo{year}{2020}\natexlab{}.
\newblock \showarticletitle{{Regular Path Query Evaluation on Streaming
  Graphs}}. In \bibinfo{booktitle}{\emph{SIGMOD International Conference on
  Management of Data}}. \bibinfo{publisher}{{ACM}},
  \bibinfo{pages}{1415--1430}.
\newblock


\bibitem[\protect\citeauthoryear{Seufert, Anand, Bedathur, and Weikum}{Seufert
  et~al\mbox{.}}{2013}]%
        {SeufertABW13}
\bibfield{author}{\bibinfo{person}{Stephan Seufert}, \bibinfo{person}{Avishek
  Anand}, \bibinfo{person}{Srikanta~J. Bedathur}, {and}
  \bibinfo{person}{Gerhard Weikum}.} \bibinfo{year}{2013}\natexlab{}.
\newblock \showarticletitle{{FERRARI: Flexible and efficient reachability range
  assignment for graph indexing}}. In \bibinfo{booktitle}{\emph{International
  Conference on Data Engineering (ICDE)}}. \bibinfo{publisher}{{IEEE} Computer
  Society}, \bibinfo{pages}{1009--1020}.
\newblock


\bibitem[\protect\citeauthoryear{Tetzel, Lehner, and Kasperovics}{Tetzel
  et~al\mbox{.}}{2020}]%
        {TetzelLK20}
\bibfield{author}{\bibinfo{person}{Frank Tetzel}, \bibinfo{person}{Wolfgang
  Lehner}, {and} \bibinfo{person}{Romans Kasperovics}.}
  \bibinfo{year}{2020}\natexlab{}.
\newblock \showarticletitle{{Efficient Compilation of Regular Path Queries}}.
\newblock \bibinfo{journal}{\emph{Datenbank-Spektrum}} \bibinfo{volume}{20},
  \bibinfo{number}{3} (\bibinfo{year}{2020}), \bibinfo{pages}{243--259}.
\newblock


\bibitem[\protect\citeauthoryear{Thompson, Personick, and Cutcher}{Thompson
  et~al\mbox{.}}{2014}]%
        {Blazegraph}
\bibfield{author}{\bibinfo{person}{Bryan~B. Thompson}, \bibinfo{person}{Mike
  Personick}, {and} \bibinfo{person}{Martyn Cutcher}.}
  \bibinfo{year}{2014}\natexlab{}.
\newblock \showarticletitle{{The Bigdata\textregistered RDF Graph Database}}.
\newblock In \bibinfo{booktitle}{\emph{Linked Data Management}}.
  \bibinfo{publisher}{Chapman and Hall/CRC}, \bibinfo{pages}{193--237}.
\newblock


\bibitem[\protect\citeauthoryear{van Rest, Hong, Kim, Meng, and Chafi}{van Rest
  et~al\mbox{.}}{2016}]%
        {RestHKMC16}
\bibfield{author}{\bibinfo{person}{Oskar van Rest}, \bibinfo{person}{Sungpack
  Hong}, \bibinfo{person}{Jinha Kim}, \bibinfo{person}{Xuming Meng}, {and}
  \bibinfo{person}{Hassan Chafi}.} \bibinfo{year}{2016}\natexlab{}.
\newblock \showarticletitle{{PGQL:} a property graph query language}. In
  \bibinfo{booktitle}{\emph{International Workshop on Graph Data Management:
  Experiences and Systems (GRADES)}}. \bibinfo{publisher}{{ACM}},
  \bibinfo{pages}{7}.
\newblock


\bibitem[\protect\citeauthoryear{Veldhuizen}{Veldhuizen}{2014}]%
        {Vel14}
\bibfield{author}{\bibinfo{person}{Todd~L. Veldhuizen}.}
  \bibinfo{year}{2014}\natexlab{}.
\newblock \showarticletitle{Triejoin: A simple, worst-case optimal join
  algorithm}. In \bibinfo{booktitle}{\emph{Proc. International Conference on
  Database Theory (ICDT)}}. \bibinfo{pages}{96--106}.
\newblock


\bibitem[\protect\citeauthoryear{Vrandecic and Kr{\"{o}}tzsch}{Vrandecic and
  Kr{\"{o}}tzsch}{2014}]%
        {VrandecicK14}
\bibfield{author}{\bibinfo{person}{Denny Vrandecic} {and}
  \bibinfo{person}{Markus Kr{\"{o}}tzsch}.} \bibinfo{year}{2014}\natexlab{}.
\newblock \showarticletitle{Wikidata: A free collaborative knowledgebase}.
\newblock \bibinfo{journal}{\emph{{Communications of the ACM}}}
  \bibinfo{volume}{57}, \bibinfo{number}{10} (\bibinfo{year}{2014}),
  \bibinfo{pages}{78--85}.
\newblock


\bibitem[\protect\citeauthoryear{Wadhwa, Prasad, Ranu, Bagchi, and
  Bedathur}{Wadhwa et~al\mbox{.}}{2019}]%
        {WadhwaPRBB19}
\bibfield{author}{\bibinfo{person}{Sarisht Wadhwa}, \bibinfo{person}{Anagh
  Prasad}, \bibinfo{person}{Sayan Ranu}, \bibinfo{person}{Amitabha Bagchi},
  {and} \bibinfo{person}{Srikanta Bedathur}.} \bibinfo{year}{2019}\natexlab{}.
\newblock \showarticletitle{{Efficiently Answering Regular Simple Path Queries
  on Large Labeled Networks}}. In \bibinfo{booktitle}{\emph{SIGMOD
  International Conference on Management of Data}}. \bibinfo{publisher}{{ACM}},
  \bibinfo{pages}{1463--1480}.
\newblock


\bibitem[\protect\citeauthoryear{Wang, Rao, Jiang, Lyu, Yang, and Feng}{Wang
  et~al\mbox{.}}{2014}]%
        {WangRJLYF14}
\bibfield{author}{\bibinfo{person}{Xin Wang}, \bibinfo{person}{Guozheng Rao},
  \bibinfo{person}{Longxiang Jiang}, \bibinfo{person}{Xuedong Lyu},
  \bibinfo{person}{Yajun Yang}, {and} \bibinfo{person}{Zhiyong Feng}.}
  \bibinfo{year}{2014}\natexlab{}.
\newblock \showarticletitle{{TraPath: Fast Regular Path Query Evaluation on
  Large-Scale RDF Graphs}}. In \bibinfo{booktitle}{\emph{Web-Age Information
  Management (WAIM)}} \emph{(\bibinfo{series}{LNCS})},
  Vol.~\bibinfo{volume}{8485}. \bibinfo{publisher}{Springer},
  \bibinfo{pages}{372--383}.
\newblock


\bibitem[\protect\citeauthoryear{Wang, Wang, and Zhang}{Wang
  et~al\mbox{.}}{2016}]%
        {WangWZ16}
\bibfield{author}{\bibinfo{person}{Xin Wang}, \bibinfo{person}{Junhu Wang},
  {and} \bibinfo{person}{Xiaowang Zhang}.} \bibinfo{year}{2016}\natexlab{}.
\newblock \showarticletitle{{Efficient Distributed Regular Path Queries on
  {RDF} Graphs Using Partial Evaluation}}. In
  \bibinfo{booktitle}{\emph{International Conference on Information and
  Knowledge Management (CIKM)}}. \bibinfo{publisher}{{ACM}},
  \bibinfo{pages}{1933--1936}.
\newblock


\bibitem[\protect\citeauthoryear{Yakovets, Godfrey, and Gryz}{Yakovets
  et~al\mbox{.}}{2013}]%
        {YakovetsGG13}
\bibfield{author}{\bibinfo{person}{Nikolay Yakovets}, \bibinfo{person}{Parke
  Godfrey}, {and} \bibinfo{person}{Jarek Gryz}.}
  \bibinfo{year}{2013}\natexlab{}.
\newblock \showarticletitle{{Evaluation of SPARQL Property Paths via Recursive
  SQL}}. In \bibinfo{booktitle}{\emph{Alberto Mendelzon International Workshop
  on Foundations of Data Management (AMW)}} \emph{(\bibinfo{series}{{CEUR}
  Workshop Proceedings})}, Vol.~\bibinfo{volume}{1087}.
  \bibinfo{publisher}{CEUR-WS.org}.
\newblock


\bibitem[\protect\citeauthoryear{Yakovets, Godfrey, and Gryz}{Yakovets
  et~al\mbox{.}}{2016}]%
        {YakovetsGG16}
\bibfield{author}{\bibinfo{person}{Nikolay Yakovets}, \bibinfo{person}{Parke
  Godfrey}, {and} \bibinfo{person}{Jarek Gryz}.}
  \bibinfo{year}{2016}\natexlab{}.
\newblock \showarticletitle{{Query Planning for Evaluating SPARQL Property
  Paths}}. In \bibinfo{booktitle}{\emph{SIGMOD International Conference on
  Management of Data}}. \bibinfo{publisher}{{ACM}},
  \bibinfo{pages}{1875--1889}.
\newblock


\end{thebibliography}

\end{document}